\documentclass[review]{elsart}

\usepackage{units} 
\usepackage{graphicx}

\usepackage[monochrome]{color}
\usepackage[disable]{todonotes}

\usepackage{xy}
\xyoption{all}

\usepackage{lineno,hyperref}
\modulolinenumbers[5]

\usepackage{amsmath}
\usepackage{amssymb}
\usepackage{stackrel}

\newtheorem{definition}{Definition}[section]
\newtheorem{theorem}[definition]{Theorem}
\newtheorem{proposition}[definition]{Proposition}

\newtheorem{example}[definition]{Example} 

\newenvironment{proof}{%
\par
\noindent
\textsc{Proof. }
\noindent}{\hfill\(\qed\)}

\newenvironment{proofsummary}{%
\par
\noindent
\textsc{Proof (Sketch). }
\noindent}{\hfill\(\qed\)}

\usepackage{environ} 
\NewEnviron{killcontents}{}

\newenvironment{append}{}{}

\newcommand\blfootnote[1]{%
  \begingroup
  \renewcommand\thefootnote{}\footnote{#1}%
  \addtocounter{footnote}{-1}%
  \endgroup
}

\newcommand{\Nat}{{\mathbb N}}
\newcommand{\Real}{{\mathbb R}}
\newcommand{\D}{\mathbb{D}}
\renewcommand{\H}{\mathbb{H}}
\newcommand{\J}{\mathbb{J}}
\renewcommand{\P}{\mathbb{P}}

\newcommand{\Hip}{\overline{\mathbb{H}}_1}
\newcommand{\Jip}{\mathbb{D}_{\le1}}
\newcommand{\conv}{\mathrm{conv}}

\newcommand{\pto}{\rightharpoondown}
\newcommand{\dom}{\mathrm{dom}}

\newcommand{\EE}{\mathcal{E}}
\newcommand{\FF}{\mathcal{F}}
\newcommand{\GG}{\mathcal{G}}
\newcommand{\TT}{\mathcal{T}}
\newcommand{\exit}{\mathbf{\Phi}}
\newcommand{\init}{\mathbf{in}}
\newcommand{\sched}{\mathbf{Sched}}
\newcommand{\unity}{\delta}
\newcommand{\sem}[1]{[\![#1]\!]}

\newcommand{\seq}{\cdot}
\newcommand{\triple}[3]{\{#1\}#2\{#3\}}
\newcommand{\pc}[1]{{\oplus_{\!#1}}}
\newcommand{\convol}{\star}
\newcommand{\bks}{*}

\newcommand{\refby}{\sqsubseteq}
\newcommand{\refbyh}{\sqsubseteq_{\mathbb{H}}}
\newcommand{\simref}{\refby_{\mathrm{sim}}}

\newcommand{\thread}{\mathtt{thd}}

\newcommand{\Eqn}[1]{Eqn.~(\ref{#1})}
\newcommand{\Prop}[1]{Prop.~\ref{#1}}
\newcommand{\Thm}[1]{Thm.~\ref{#1}}
\newcommand{\Defs}[1]{Def. #1}
\newcommand{\Sec}[1]{Sec. \ref{#1}}
\newcommand{\Fig}[1]{Fig.~\ref{#1}}

\newcommand{\Gx}{\color{blue}}
\newcommand{\Tx}{\color{orange}}

\let\OLDthebibliography\thebibliography
\renewcommand\thebibliography[1]{
  \OLDthebibliography{#1}
  \setlength{\parskip}{0pt}
  \setlength{\itemsep}{0pt plus 0.3ex}
}

\begin{document}

\begin{frontmatter}
\title{Probabilistic Rely-guarantee Calculus}

\author[mq]{Annabelle McIver}
\ead{annabelle.mciver@mq.edu.au}
\author[mq]{Tahiry Rabehaja}
\ead{tahiry.rabehaja@mq.edu.au}
\author[sh]{Georg Struth}
\blfootnote{This research was supported by an iMQRES from Macquarie University, the ARC Discovery Grant DP1092464 and the EPSRC Grant EP/J003727/1.}
\ead{g.struth@sheffield.ac.uk}
\address[mq]{Department of Computing, Macquarie University, Australia}
\address[sh]{Department of Computer Science, University of Sheffield, United Kingdom}

\begin{abstract} Jones' rely-guarantee calculus for shared variable concurrency is extended to include probabilistic behaviours. We use an algebraic approach that is based on a combination of probabilistic Kleene algebra with concurrent Kleene algebra.  \todo[disable]{adapts with concurrent KA sounds odd. Why not write: We use an algebraic approach that is based on a combination of probabilistic Kleene algebra with concurrent Kleene algebra?} Soundness of the algebra is shown relative to a general probabilistic event structure semantics. The main contribution of this paper is a collection of rely-guarantee rules built on top of that semantics. In particular, we show how to obtain bounds on probabilities of correctness by deriving quantitative extensions of rely-guarantee rules. 
\todo[inline]{I don't fully understand the previous sentence: probability bounds should come from the probabilistic part of the semantics, not from the concurrent one... {\Tx Does that make more sense?}} 
The use of these rules is illustrated by a detailed verification of a simple probabilistic concurrent program: a faulty Eratosthenes sieve.  \end{abstract}


\end{frontmatter}


\section{Introduction}
The rigorous study of concurrent systems remains a difficult task due to the intricate interactions {\Gx and interferences} between their components. A formal framework for concurrent systems ultimately depends on the kind of concurrency considered. Jones' rely-guarantee calculus provides a mathematical foundation for {\Gx proving} the correctness of programs with shared variables concurrency{\Gx in compositional fashion}~\cite{Jon81}. This paper extends Jones' calculus to the quantitative correctness of probabilistic concurrent programs.

Probabilistic programs have become popular due to their ability to express quantitative rather than limited qualitative properties. Probabilities are particularly important for protocols that rely on the unpredictability of probabilistic choices.
\todo[disable]{protocols requiring security?} 
The sequential probabilistic semantics, originating with Kozen~\cite{Koz81} and Jones~\cite{Jon92}, have been extended with nondeterminism~\cite{He97,Mci04}, to yield methods for quantitative reasoning based on partial orders. \todo[inline]{What does using partial orders mean? {\Tx Rewritten}} 

{\Gx
We aim to obtain similar methods for reasoning in compositional ways about probabilistic programs with shared variable concurrency. In algebraic approaches, compositionality arises quite naturally through congruence or monotonicity properties of algebraic operations such as sequential and concurrent composition or probabilistic choice.

It is well known that compositional reasoning is nontrivial both for concurrent and for sequential probabilistic systems.  In the concurrent case, the obvious source of non-compositionality is communication or interaction between components.  In the rely-guarantee approach, interference conditions are imposed between individual components and their environment in order to achieve compositionality. Rely conditions account for the global effect of the environment's interference with a component; guarantee conditions express the effect of a particular component on the environment.  Compositionality is then obtained by considering rely conditions within components and guarantee conditions within the environment.
}

In the presence of probabilistic behaviours, a problem of congruence (and hence non-compositionality) arises when considering the natural extension of trace-based semantics to probabilistic automata~\cite{Seg94}, where a standard work-around is to define a partial order based on simulations.
\todo[inline]{I would prefer to see a conceptual explanation instead of a mathematical one, as in the r-g case.}

In this paper, we define a similar construct to achieve compositionality. However, simulation-based equivalences are usually too discriminating for program verification. Therefore, we also use a weaker semantics that is essentially based on sequential behaviours. Such a technique has been motivated elsewhere~\cite{Arm14}, where the sequential order is usually not a congruence. Therefore, the simulation-based order is used for properties requiring composition while the second order provides a tool that captures the sequential behaviours of the system.

Concurrent Kleene algebra~\cite{Hoa11,Arm14} provides an algebraic account of Jones' rely-guarantee framework. Algebras provide an abstract view of a program by focusing more on control flows rather than data flows. All the rely-guarantee rules described in~\cite{Hoa11,Arm14} were derived by equational reasoning from a finite set of algebraic axioms. Often, the verification of these axioms on an intended semantics is easier than proving the inference rules directly in that semantics. Moreover, every structure satisfying these laws will automatically incorporate a direct interpretation of the rely-guarantee rules, as well as additional rules that can be used for program refinement. Therefore, we also adopt an algebraic approach to the quantitative extension of rely-guarantee, that is, we establish some basic algebraic properties of a concrete event structure model and derive the rely-guarantee rules by algebraic reasoning.

In summary, the main contribution of this paper is the development of a mathematical foundation for \textit{probabilistic rely-guarantee calculi}. The inference rules are expressed algebraically, and we illustrate their use on an example based on the Sieve of Eratosthenes which incorporates a probability of failure. We also outline two rules that provide probabilistic lower bounds for the correctness of the concurrent execution of multiple components. 

A short summary of the algebraic approach to rely-guarantee calculus and the extension to probabilistic programs are found respectively in Section~\ref{sec:standard-rg} and~\ref{sec:prgc}-\ref{sec:prg-rules}. Section~\ref{sec:sequential-programs} and~\ref{sec:es} are devoted to the construction of a denotational model for probabilistic concurrent programs.
Section~\ref{sec:application} closes this paper with a detailed verification of the faulty Eratosthenes sieve. 



\section{Non-probabilistic rely-guarantee calculus}\label{sec:standard-rg}
The rely-guarantee approach, originally put forward by Jones~\cite{Jon81}, is a compositional method for developing and verifying large concurrent systems. An algebraic formulation of the approach has been proposed recently in the context of concurrent Kleene algebras~\cite{Hoa11}.  In a nutshell, a bi-Kleene algebra is an algebraic structure $(K,+,\cdot,\|,0,1,^\ast,^{(\ast)})$ such that $(K,+,\cdot,0,1,^\ast)$ is a Kleene algebra and $(K,+,\|,0,1,^{(\ast)})$ is a commutative Kleene algebra. The axioms of Kleene algebra and related structures are in Appendix~\ref{A:ka} 
.

\todo[inline]{I think such an appendix should be added to summarise the algebraic laws used in this paper.}

Intuitively, the set $K$ models the actions a system can take; the operation $(+)$ corresponds to the nondeterministic choice between actions, $(\cdot)$ to their sequential composition and $(\|)$ to their parallel or concurrent composition. The constant $0$, the unit of addition, models the abortive action, $1$, the unit of sequential and concurrent composition, the ineffective action $\mathtt{skip}$. The operation $(^\ast)$ is a sequential finite iteration of actions; the corresponding parallel finite iteration operation $(^{(\ast)})$ is not considered further in this article. Two standard models of bi-Kleene algebras are languages, with $(+)$ interpreted as language union, $(\cdot)$ as language product, $(\|)$ as shuffle product, $0$ as the empty language, $1$ as the empty word language and $(^\ast)$ as the Kleene star, and pomset languages under operations similarly to those in Section~\ref{S:seqred} below (cf. ~\cite{BloomEsik}).

Language-style models with interleaving or shuffle also form the standard semantics of rely-guarantee calculi. In that context, traces are typically of the form $(s_1,s_1'),(s_2,s_2')\dots (s_k,s_k')$, where the $s_i$ and $s_i'$ denote states of a system, pairs $(s_i,s_i')$ correspond to internal transitions of a component, and fragments $s_i'),(s_{i{+}1}$ to transitions caused by interferences of the environment. Behaviours of a concurrent system are associated with sets of such traces.

With semantics for concurrency in mind, a generalised encoding of the validity of Hoare triples becomes useful:  
\begin{equation*}
  \triple{P}{S}{Q} \Leftrightarrow P{\cdot} S \le Q,
\end{equation*}
where $P\le Q \Leftrightarrow P{\cup} Q = Q$.  It has been proposed originally by Tarlecki~\cite{Tar85} for sequential programs with a relational semantics. In contrast to Hoare's standard approach, where $P$ and $Q$ are assertions and $S$ a program, all three elements are now allowed to be programs. In the context of traces, $\triple{P}{S}{Q}$ holds if all traces that are initially in $P$ and then in $S$ are also in $Q$. This comprises situations where program $P$ models traces ending in a set of states $p$ (a precondition) and $Q$ models traces ending in a set of states $q$ (a postcondition). The Hoare triple then holds if all traces establishing precondition $p$ can be extended by program $S$ to traces establishing postcondition $q$, whenever $y$ terminates, as in the standard interpretation. We freely write $\triple{p}{S}{q}$ in such cases. It turns out that all the inference rules of Hoare logic except the assignment rule can be derived in the setting of Kleene algebra~\cite{Hoa11}.

For concurrency applications, the algebraic encoding of Hoare triples has been expanded to Jones quintuples $\triple{P\ R}{S}{G\ Q}$, also written $R,G\vdash\triple{P}{S}{Q}$, with respect to rely conditions $R$ and guarantee conditions $G$~\cite{Hoa11}. The basic intuition is as follows. A rely condition $R$ is understood as a special program that constrains the behaviour of a component $S$ by executing it in parallel as $R\| S$. This is consistent with the above trace interpretation where parallel composition is interpreted as shuffle and gaps in traces correspond to interferences by the environment. Typical properties of relies are $1\le R$ (where $1$ is $\mathtt{skip}$) and $R^\ast = R{\cdot} R= R\| R = R$. Moreover, relies distribute over nondeterministic choices as well as sequential and concurrent compositions: $R\| (S{+}T)= R\|S {+} R\|T$, $R\| (S{\cdot} T)=(R\| S) {\cdot} (R\| T)$ and $R\|(S\| T)=(R\| S)\| (R\| T)$, hence they apply to all subcomponents of a given component~\cite{Arm14}. A guarantee $G$ of a given component $S$ is only constrained by the fact that it should include all behaviours of $S$, that is, $S\le G$.

\todo[inline]{\Tx I wonder if using the word constrain both for the guarantee (previous sentence) and the rely (next sentence) wouldn't introduce confusion.}

Consequently, a Jones quintuple is valid if the component $S$ constrained by the rely satisfies the Hoare triple---the relationship between precondition and postcondition---and the guarantee includes all behaviours of $S$~\cite{Hoa11}: \begin{equation}{\label{eq:rgspec}}
  \triple{P\ R}{S}{G\ Q} \Leftrightarrow \triple{P}{R\| S}{Q} \wedge S\le G.
\end{equation}
The rules of Hoare logic without the assignment axiom are still derivable from the axioms of bi-Kleene algebra, when Hoare triples are replaced by Jones quintuples~\cite{Hoa11}. To derive the standard rely-guarantee concurrency rule, one can expand bi-Kleene algebra by a meet operation $(\sqcap)$ and assume that $(K,+,\sqcap)$ forms a distributive lattice~\cite{Arm14}. Then
\begin{equation}\label{rule:rg-standard}
\frac{\triple{P\ R}{S}{G\ Q}\quad\triple{P\ R'}{S'}{G'\ Q'}\quad G\leq R'\quad G'\leq R}{\triple{P\ R{\sqcap} R'}{S\|S'}{G{+} G'\ Q{\sqcap} Q',}}. 
\end{equation}
This inference rule demonstrates how the rely-guarantee specifications of components can be composed into a rely-guarantee specification of a larger system. If $S$ and $S'$ satisfy the premises, then $S\| S'$ satisfies both postconditions $Q$ and $Q'$ when run in an environment satisfying both relies $R$ and $R'$. Moreover, $S\| S'$ guarantees either of $G$ or $G'$.

Deriving these inference rules from the algebraic axioms mentioned makes them sound with respect to all models of these axioms, including trace-based semantics with parallel composition interpreted as interleaving, and true-concurrency semantics such as pomset languages and the event structures considered in this article. Without the algebraic layer, Dingel~\cite{Din02} and Coleman and Jones~\cite{Col07} have already proved the soundness of rely-guarantee rules with respect to trace-based semantics, more precisely \emph{Aczel traces}~\cite{Roe97}.  This paper follows previous algebraic developments, but for probabilistic programs.

In Section~\ref{sec:prgc}, we provide a suitable extension of the rely-guarantee formalism, in particular Rule~(\ref{rule:rg-standard}), to probabilistic concurrent programs. The soundness of such a formalism is shown relative to a semantic space that allows sequential probabilistic programs to include concurrent behaviours.

\section{Sequential probabilistic programs}\label{sec:sequential-programs}
We start by giving a brief summary of the denotation of sequential probabilistic programs using the powerdomain construction of McIver and Morgan~\cite{Mci04}. All probabilistic programs are considered to have a finite state space denoted by $\Omega$. A distribution over the set $\Omega$ is a function $\mu{:}\Omega{\to}[0,1]$ such that $\sum_{s{\in}\Omega}\mu(s) {=} 1$. The set of distributions over $\Omega$ is denoted by $\D\Omega$. Since $\Omega$ is a finite set, we identify a distribution with the associated measure. For every $\mu{\in}\D\Omega$ and $O{\subseteq}\Omega$, we write $\mu(O) {=} \sum_{s{\in} O}\mu(s)$. An example of distribution is the point distribution $\delta_s$, centred at the state $s{\in}\Omega$, such that
\begin{displaymath}
\delta_s(s') = \begin{cases}
1 & \textrm{ if } s {=} s',\\
0 & \textrm{otherwise.}
\end{cases}
\end{displaymath}
A (nondeterministic) probabilistic program $r$ modelled as a map of type $\Omega{\to}\P\D \Omega$ such that $r(s)$ is a non-empty, topologically closed and convex subset of $\D \Omega$ for every state $s{\in} \Omega$. The set $\D\Omega$ is a topological sub-space of the finite product $\Real^\Omega$ (endowed with the usual product topology), and the topological closure is considered with respect to the induced topology on $\Omega$\footnote{These healthiness conditions are set out and fully explained in the work of McIver and Morgan~\cite{Mci04}.}. We denote by $\H_1 \Omega $ the set of probabilistic programs that terminate almost certainly. Notice that the set $\D\Omega$ contains only distributions instead of the subdistributions considered by McIver and Morgan~\cite{Mci04}. Therefore, we only model nondeterministic programs that are terminating with probability $1$.

Programs in $\H_1 \Omega$ are ordered by pointwise inclusion, i.e. $r\refbyh r'$ if for every $s{\in} \Omega$, $r(s)\subseteq r(s')$. A program $r$ is deterministic if, for every $s$, $r(s) = \{\mu_s\}$ (i.e. a singleton) for some distribution $\mu_s{\in} \D \Omega$. The set of deterministic programs is denoted by $\J_1 \Omega$ (as in Jones' spaces~\cite{Jon92}). If $f{\in}\J_1\Omega$ is a deterministic program such that $f(s) = \{\mu_s\}$, then we usually just write $f(s) = \mu_s$. A particularly useful example of a probabilistic deterministic program is the ineffectual program $\mathtt{skip}$, which we denote by $\delta$. Thus $\delta(s) = \{\delta_s\}$.


Let $p\in[0,1]$. The probabilistic combination of two probabilistic programs $r$ and $r'$ is defined as~(\cite[\Defs{5.4.5}]{Mci04}) 
\begin{equation}\label{def:prog-probabilistic-choice}
(r\pc{p} r')(s) = \{\mu\pc{p}\mu'\ |\ \mu{\in} r(s)\wedge\mu'{\in} r'(s)\},
\end{equation}
where $(\mu \pc{p}\mu')(s) = (1{-}p)\mu(s) {+} p\mu'(s)$ for every state $s{\in}\Omega$. Thus, the program $r$ (resp. $r'$) is executed  with probability $1{-}p$ (resp. $p$).

Nondeterminism is obtained as the set of all probabilistic choices (\cite[\Defs{5.4.6}]{Mci04} ), that is, \begin{equation}\label{def:prog-nondeterminism} (r{+}r')(s) = {\cup}_{p{\in}[0,1]} (r\pc{p}r')(s).  \end{equation}

The sequential composition of $r$ by $r'$ is defined as (\cite[\Defs{5.4.7}]{Mci04}):
\begin{equation}\label{eq:6-sequential-H}
(r{\cdot} r')(s) =  \left\lbrace\left. f{\convol} \mu\right| f{\in}\J_1\Omega\wedge \mu{\in} r(s)\wedge f\refbyh r' \right\rbrace
\end{equation}
where 
\[
	(f{\convol}\mu)(s') = \sum_{s''{\in}\Omega}f(s'')(s')\mu(s'')
\]
for every state $s'{\in}\Omega$. 

%
%


For $r,r'{\in}\H_1\Omega$, the binary Kleene star $r{\bks} r'$ is the least fixed point of the function $f_{r,r'}(X) = r' {+} r{\cdot} X$ in $\H_1\Omega$. It has been shown in~\cite{Mci04} that the function $r'\mapsto r{\cdot} r'$ is continuous ---it preserves directed suprema. Notice that a topological closure is sometimes needed to ensure that we obtain an element of $\H_1\Omega$. Hence, the Kleene star $r{\bks} r'$ is the program such that
$r{\bks}r'(s) = \overline{{\cup}_{n}f_{r,r'}^n(\bot)(s)}$,
where $\overline{A}$ is the topological closure of the set $A\subseteq\D\Omega$ and the constant $\bot$ is defined, as usual, such that $r''{\cdot} \bot {=} \bot{\cdot} r'' {=} \bot$, $\bot{+}r'' = r''$ and $\bot\refbyh r''$ for every $r''{\in}\H_1\Omega{\cup}\{\bot\}$.

We introduce tests, which are used for conditional constructs, following the idea adopted in various algebras of programs. We define a test to be a map $b:\Omega\to\P\D\Omega$ such that $b(s) \subseteq\{\delta_s\}$. Indeed, an ``if statement" is modelled algebraically as $b{\cdot} r {+} (\neg b){\cdot} r'$ where $(\neg b)(s) = \emptyset$ if the test underlying $b$ holds at state $s$ and it is $\{\delta_s\}$ otherwise. The sub-expression $b{\cdot} r(s)$ still evaluates to $\emptyset$ if $b(s)$ is empty, but care should be taken to avoid expressions such as $r{\cdot} b$ (if $f$ is a deterministic refinement of $b$, then $f(s'')(s')$ may have no meaning if $b(s'') {=} \emptyset$). A test that is always false can be identified with $\bot$.

We denote by $\Hip\Omega$ the set of tests together with the set of probabilistic programs. The refinement order $\refbyh$ is extended to $\Hip\Omega$ in a straightforward manner. For every test $b$, we have $b\refbyh \delta$; hence, we refer to tests as \emph{subidentities}. Every elements of $\Hip\Omega$ are called programs, unless otherwise specified.



\section{An event structures model for probabilistic concurrent programs}\label{sec:es}

The set $\H_1\Omega$ of probabilistic programs provides a full semantics for program constructs such as (probabilistic) assignments, probabilistic choices, conditionals and while loops that terminate almost surely. Unfortunately, it is impossible to define the concurrent composition of two sequential programs as an operation on $\H_1\Omega$ because the result would always be a sequential program. Thus we are forced to look for a more general framework in order to formally model concurrency. Fortunately, there are several suitable mathematical models that allows the formal verification of programs with concurrent behaviours. A powerful example that accounts for true concurrency are Winskel's \emph{event structures}~\cite{Win80,Win86}. In this section, we outline a denotational semantics for probabilistic concurrent programs based on Langerak's bundle event structures~\cite{Lan92}, which have been extended successfully to quantitative features~\cite{Kat93,Kat96,Var03}. This construction is necessary to ensure the soundness of the extended rely-guarantee formalism.

{\Tx
A bundle event structure comprises events ranging over some set $E$ of \emph{events} as its fundamental objects. Intuitively, an event is an occurrence of an action at a certain moment in time. Thus an action can be repeated, but each of its occurrences is associated with a unique event. Events are (partially) ordered by a causality relation which we denote by $\mapsto$: if an event $e''$ causally depends on either $e$ or $e'$ (i.e. $\{e,e'\}{\mapsto} e''$) then either $e$ or $e'$ must have happened before $e''$ can happen or is \emph{enabled}. The relationship between $e$ and $e'$ is called \emph{conflict}, written $e\#e'$, because both events cannot occur simultaneously.
}

In general, the conflict relation $\#$ is a binary relation on $E$. Given two subsets $x,x'\subseteq E$, the predicate $x\#x'$ holds iff for every $(e,e'){\in} x{\times} x'$ such that $e{\neq}e'$, we have $e\#e'$. 

\begin{definition}\label{def:ipbes}
A quintuple $\EE = (E,\mapsto,\#,\lambda,\exit)$ is a \emph{bundle event structure with internal probability} (i.e. an ipBES) if 
\begin{itemize}
\item $\#$ is an irreflexive symmetric binary relation on $E$, called \emph{conflict relation}. 
\item $\mapsto\subseteq\P E{\times} E$ is a \emph{bundle relation}, i.e. if $x{\mapsto} e$ for some $x\subseteq E$ and $e{\in} E$, then $x\#x$. 
\item $\lambda{:}E{\to}\Hip\Omega$, i.e. it labels events with (atomic) probabilistic programs.
\item $\exit\subseteq \P E$ such that $x\#x$ holds  for every $x{\in}\exit$.
\end{itemize}
 
The finite state space $\Omega$ of the programs used as labels is fixed.
\end{definition} 

The intuition behind this definition is that events are occurrences of atomic program fragments, i.e. they can happen without interferences from an environment. Hence, we need to distinguish all atomic program fragments when translating a program into a bundle event structure. Atomic programs can be achieved by creating a construct that forces atomicity. Examples of such a technique include ``atomic brackets"~\cite{Jon12}. In this paper, we always state which actions are atomic rather than using such a device.

Given an ipBES $\EE$, a \emph{finite trace} of $\EE$ is a sequence of events $e_1e_2\dots e_n$ such that for all different $1\leq i,j\leq n$, $\neg (e_i\#e_j)$ and if $j {=} i{+}1$ then there exists an $x\subseteq E$ such that $x{\mapsto} e_j$ and $e_i{\in} x$~\cite{Lan92,Kat96,Rab13}. In other words, a trace is safe (an event may occur only when it is enabled) and is conflict free. The set of all finite traces of $\EE$ is denoted by $\TT(\EE)$. The set of maximal traces of $\EE$ (w.r.t the prefix ordering) is denoted $\TT_{\max}(\EE)$. We simply write $\TT$ (resp. $\TT_{\max}$) instead of $\TT(\EE)$ (resp. $\TT_{\max}(\EE)$) when no confusion may arise.

The aim of this section is to elaborate two relationships between the sets of traces of given event structures. The first comparison is based on a sequential reduction using schedulers; the second one is simulation. We will show that the sequential comparison is strictly weaker than the simulation relation.

\subsection{Schedulers on ipBES}\label{sec:scheduler}

{\Tx
As in the case of automata, we define schedulers on ipBES in order to obtain a sequential equivalence on bundle event structures with internal probability. Intuitively, a scheduler reduces an ipBES to a element of $\Hip\Omega$. While the technicalities of the schedulers we define in this paper is tailored towards a rely-guarantee reasoning, there might be relationships with previous works~\cite{Geo10,Geo12} where schedulers (and associated testing theories) are restricted in order achieve a broader class observationally equivalent processes. 
}

A \emph{subdistribution} is a map $\mu:\Omega\to[0,1]$ such that $\sum_{s{\in}\Omega}\mu(s)\leq 1$.  The set of subdistributions over $\Omega$ is denoted by $\Jip\Omega$.  

\begin{definition}\label{def:ipscheduler}
A \emph{scheduler} $\sigma$ on an ipBES $\EE$ is a map 
\[
	\sigma{:}\TT{\to} [(E{\times} \Omega){\pto} \Jip\Omega]
\] 
such that for all $\alpha{\in}\TT$:
\begin{enumerate}
\item $\dom(\sigma(\alpha)) = \{(e,s)\ |\ \alpha e{\in}\TT\wedge s{\in}\Omega\}$,\label{pr:sched-dom}
\item there exists a function $w{:}E{\times}\Omega{\to}[0,1]$ such that, for every $(e,s){\in}\dom(\sigma(\alpha))$, $\sigma(\alpha)(e,s) = w(e,s)\mu$ for some $\mu{\in}\lambda(e)(s)$.
\label{pr:sched-choice}
\item for every $ s{\in}\Omega$, we have $\sum_{(e,s){\in}\dom(\sigma(\alpha))} w(e,s) = 1$,\label{pr:sched-prob}
\item for every $(e,s){\in}\dom(\sigma(\alpha))$, if $\lambda(e)(s) {=} \emptyset$, then  $w(e,s) {=} 0$ and $\sigma(\alpha)(e,s) {=} 0$ (the subdistribution that evaluates to $0$ everywhere).\label{pr:sched-consistent}
\end{enumerate}
The set of all schedulers on $\EE$ is denoted by $\sched(\EE)$.
\end{definition}

Property~\ref{pr:sched-dom} says that we may schedule an event provided it does not depend on unscheduled events.

Property~\ref{pr:sched-choice} states that, given a trace $\alpha$, the scheduler will resolve the nondeterminism between events enabled after $\alpha$ byusing the weight function $w$. This may include immediate conflicts or interleavings of concurrent events. Moreover, the scheduler has access to the current program state when resolving that nondeterminism. This means that $w(e,s)$ is the probability that the event $e$ is scheduled, knowing that the program state is $s$. If the event $e$ is successfully scheduled, then the scheduler performs a last choice of distribution, say $\mu$ from $\lambda(e)(s)$, to generate the next state of the program.

Property~\ref{pr:sched-prob} ensures that when the state $s$ is known, then the choice between the events, enabled after the trace $\alpha$, is indeed probabilistic.  

Property~\ref{pr:sched-consistent} says that a scheduler is forced to choose events whose labels do not evaluate to the empty set at the current state of the program. This is particularly important when the program contains conditionals and the label of an event is a test. A scheduler is forced to choose the branch whose test holds. If two tests hold at state $s$, then a branch is chosen probabilistically using the weight function $w$. 

The motivation behind Property~\ref{pr:sched-consistent} is to ensure that, for every trace $\alpha$ such that $\dom(\sigma(\alpha)){\neq}\emptyset$, and every state $s{\in}\Omega$, we have
\[
	\sum_{(e,s){\in} \dom(\sigma(\alpha))} \sigma(\alpha)(e,s) {\in}\D\Omega,
\]
hence that sum is indeed a distribution. To ensure that a scheduler satisfying that condition can be constructed, we restrict ourselves to \textit{feasible} event structures. Given an element $r{\in}\Hip\Omega$, we write $\dom(r) {=} \{s \ |\ r(s){\neq}\emptyset\}$.

\begin{definition}
An ipBES $\EE$ is \emph{feasible} if for every trace $\alpha {\in}\TT{\setminus}\TT_{\max}$, we have ${\cup}_{\alpha e{\in}\TT}\dom(\lambda(e)) {=} \Omega$.
\end{definition}

A consequence of this assumption is that an ``if clause" always needs to have a corresponding ``else clause".

\begin{example}
Let us consider the program $r{\cdot}(\unity{+}r)$. In this program, $r$ is  atomic deterministic (such as an assignment to a variable) and the associated event structure has three events:
\[
	\EE = (\{e_{r},e_{r}',e_\delta\},\{\{e_{r}\}{\mapsto} e_\delta,\{e_{r}\}{\mapsto} e_{r}'\},\{e_\delta\#e_{r_2}\},\{(e_r,r),(e_r',r),(e_\delta,\delta)\},\exit),
\]
where $\exit = \{\{e_r',e_\delta\}\}$ (see \Sec{S:seqred} for an inductive construction of ipBES from primitive blocks). This event structure is feasible and a scheduler $\sigma$ on $\EE$ is characterised by a weight function $w{:}\{e_r,e_r',e_\delta\}{\times}\Omega{\to}[0,1]$ resolving the choice $\unity {+} r$.  In fact, for every fixed state $s{\in}\Omega$, we have $\sigma(e_r)(e_\delta,s){=}w(e_\delta,s)\delta_s$ and $\sigma(e_r)(e_r',s){=}w(e_r',s)r(s)$ and $w(e_\delta,s) {+} w(e_r',s) {=} 1$.
\end{example}

\subsection{Generating sequential probabilistic programs from ipBES and schedulers}

Similar to the case of probabilistic automata~\cite{Seg94}, our scheduler resolves branching as encoded in the conflict relation of an event structure. In addition, a scheduler also ``flattens" concurrency into interleaving by choosing an enabled event according to the associated weight function. The flattening of concurrent behaviours is sound because actions labelling events are assumed atomic and we are using schedulers to generate sequential behaviours from an ipBES. True concurrency is accounted for in \Sec{s1511}.

Let $\sigma{\in}\sched(\EE)$ and $s{\in}\Omega$ be an initial state. We inductively construct a sequence of functions $\varphi_n$ that map a trace in $\TT$ to a subdistribution on $\Omega$ according to $\sigma$ and $s$. Intuitively, if $\alpha{\in}\TT$, then $\varphi_n(\alpha){\in}\Jip\Omega$ is the sequential composition of the $n$-first probabilistic actions labelling events in $\alpha$ applied to the initial state $s$. This yields a subdistribution because $\alpha$ is weighted with respect to the scheduler $\sigma$. The sequence of partial functions $\varphi_n{:}\TT{\pto}\Jip\Omega$ is the \emph{computation sequence} of $\EE$ with respect to $\sigma$  from initial state $s$.

Formally, for each $n{\in}\Nat$, we have $\dom(\varphi_n) {=} {\cup}_{k\leq n}\TT_k$, where $\TT_n$ is the set of traces of length $n$ and 

\begin{enumerate}
\item $\varphi_0(\emptyset) {=} \delta_s$,where $s$ is the initial state,
\item if $\alpha e{\in}\TT_{n{+}1}$ then 
\[
	\varphi_{n{+}1}(\alpha e)(s) = \sum_{t{\in}\Omega}[\sigma(\alpha)(e,t)(s)]\varphi_n(\alpha)(t)
\]
and  $\varphi_{n{+}1}(\alpha e) {=} \varphi_n(\alpha e)$ otherwise.\label{pr:induction-computation-function}
\end{enumerate}

To emphasises that this computation function refers to a specific initial state $t{\in}\Omega$ we sometimes write $\varphi_{n,t}$ instead of $\varphi_n$.

The \textit{complete run} of $\EE$ with respect to $\sigma$ is the limit $\varphi$ of that sequence, i.e. $\varphi {=} {\cup}_n\varphi_n$, which exists because $\varphi_n$ defines a sequence of partial functions such that $\varphi_{n}$ is the restriction of $\varphi_{n{+}1}$ to $\dom(\varphi_n)$. Since we consider finite traces only, we have $\dom(\varphi) {=} \TT$. The \textit{sequential behaviour} of $\EE$ with respect to $\sigma$ from the initial state $s$ is defined by the sum \[
	\sigma_s(\EE) = \sum_{\alpha{\in}\TT_{\max}}\varphi(\alpha).
\]

\begin{proposition}\label{pro:scheduler-well-defined}
For every bundle event structure $\EE$, scheduler $\sigma{\in}\sched(\EE)$ and initial state $s$, $\sigma_s(\EE)$ is a subdistribution. 
\end{proposition}

\begin{proofsummary}
The proof is by induction on the sequence of computation functions. We can show by induction on $n$ that 
\[
	\mu_{n}(\Omega) = \sum_{\alpha{\in} \TT_{n}{\cup}(\TT_{\max}{\cap} \dom(\varphi_n))}\varphi(\alpha)(\Omega) = \sum_{t{\in}\Omega} \sum_{\alpha{\in} \TT_{n}{\cup}(\TT_{\max}{\cap} \dom(\varphi_n))}\varphi(\alpha)(t) = 1
\]
and deduce that, at the limit, $\sum_{\alpha{\in}\TT_{\max}}\varphi(\alpha)(\Omega) \leq 1$.
\end{proofsummary}

\begin{proof}
Let $\varphi$ be the complete run of $\EE$ with respect to a given scheduler $\sigma$. We show by induction on $n$ that
\[
	\mu_{n}(\Omega) = \sum_{\alpha{\in} \TT_{n}{\cup}(\TT_{\max}{\cap} \dom(\varphi_n))}\varphi(\alpha)(\Omega) = \sum_{t{\in}\Omega} \sum_{\alpha{\in} \TT_{n}{\cup}(\TT_{\max}{\cap} \dom(\varphi_n))}\varphi(\alpha)(t) = 1.
\]
For the base case $n=0$, we have $\mu_0(\Omega) = \varphi(\emptyset)(\Omega) = \delta_s(\Omega) = 1$, where $s$ is the initial state. Assume the induction hypothesis $\mu_n(\Omega) = 1$. We have
\begin{align*}
\mu_{n{+}1}(\Omega) & =   \sum_{\alpha{\in} \TT_{n{+}1}{\cup}(\TT_{\max}{\cap} \dom(\varphi_{n{+}1}))}\varphi(\alpha)(\Omega)\\
& =^{\dag} \sum_{\alpha{\in} \TT_{n{+}1}}\varphi(\alpha)(\Omega) {+} \sum_{\alpha{\in}\TT_{\max}{\cap} \dom(\varphi_n)}\varphi(\alpha)(\Omega)\\
& = \sum_{\alpha e{\in} \TT_{n{+}1}}\sum_{t{\in}\Omega}\sigma(\alpha)(e,t)(\Omega)\varphi(\alpha)(t) {+} \sum_{\alpha{\in} \TT_{\max}{\cap} \dom(\varphi_n)}\varphi(\alpha)(\Omega)\\
& =  \sum_{\alpha{\in} \TT_{n}{\setminus}\TT_{\max}}\sum_{\alpha e{\in} \TT}\sum_{t{\in}\Omega}\sigma(\alpha)(e,t)(\Omega)\varphi(\alpha)(t) {+} \sum_{\alpha{\in} \TT_{\max}{\cap} \dom(\varphi_n)}\varphi(\alpha)(\Omega)\\
& = \sum_{\alpha{\in} \TT_{n}{\setminus}\TT_{\max}}\left[\sum_{(e,t){\in}\dom(\sigma(\alpha))}\sigma(\alpha)(e,t)(\Omega)\right]\varphi(\alpha)(t) {+} \sum_{\alpha{\in}\TT_{\max}{\cap} \dom(\varphi_n)}\varphi(\alpha)(\Omega)\\
& =^{\ddag} \sum_{\alpha{\in} \TT_{n}{\setminus}\TT_{\max}}\varphi(\alpha)(\Omega) {+} \sum_{\TT_{\max}{\cap} \dom(\varphi_n)}\varphi(\alpha)(\Omega)\\
& =  \mu_n(\Omega) = 1.
\end{align*} 
($\dag$) Follows from $\TT_{n{+}1}{\cup}(\TT_{\max}{\cap} \dom(\varphi_{n{+}1})) = \TT_{n{+}1}{\cup}(\TT_{\max}{\cap} \dom(\varphi_{n}))$ and the fact that the second union is disjoint.

($\ddag$) The square-bracketed term equals $1$ because of Properties~\ref{pr:sched-choice} and~\ref{pr:sched-prob} of the scheduler $\sigma$.

Therefore, each partial computation $\varphi_n$ can be seen as a probability distribution $\varphi_n(-)(\Omega)$ supported on $\TT_n{\cup}(\TT_{\max}{\cap}\dom(\varphi_n))$. Hence, the limit is a subdistribution $\varphi(-)(\Omega)$ on $\TT_{\max}$. It does not necessarily add up to $1$ because elements of $\TT_{\max}$ are finite maximal traces only and non-termination will decrease that quantity (we assume that the empty sum is $0$. This occurs when there are no maximal traces).
\end{proof}

Given a state $t{\in}\Omega$, $\sigma_s(\EE)(t)$ is the probability that the concurrent probabilistic program denoted by $\EE$ terminates in state $t$ when conflicts (resp. concurrent events) are resolved (resp. interleaved) according to the scheduler $\sigma$. Since we consider terminating programs only, we denote by $\sched_1(\EE)$ the set of schedulers of $\EE$ such that, for every initial state $s$, $\sigma_s(\EE)$ is a distribution. A scheduler in $\sched_1(\EE)$  generates a sequential behaviour that terminates almost surely. This leads to our definition of a bracket $\sem{\ }$ that transform each feasible ipBES to an element of $\H_1\Omega$: \[
	\sem{\EE}(s) = \overline{\conv\{\sigma_s(\EE)\ |\ \sigma{\in}\sched_1(\EE)\}}
\]
where $\conv(A)$ (resp. $\overline{A}$) is the convex (resp. topological) closure of the set of distributions $A$ in $\Real^\Omega$. 

\begin{definition}\label{def:semantics-sequential}
Let $\EE,\FF$ be two feasible event structures. We say that $\EE$ (sequentially) refines $\FF$, denoted by $\EE\refby\FF$, if $\sem{\EE}\refbyh\sem{\FF}$ holds in $\H_1\Omega$.
\end{definition}

The relation $\refby$ is a preorder on ipBES. Whilst this order is not a congruence, it is used to specify the desired sequential properties of a feasible event structure $\EE$ with $\sched_1(\EE){\neq}\emptyset$. We will show that feasibility and non-emptiness of $\sched_1$ are preserved by the regular operations of the next section (Props \ref{pro:homomorphism} and \ref{pro:*-homomorphism}).

\subsection{Regular operations on ipBES}\label{S:seqred}

This section provides interpretations of the operations $(+,\cdot,\bks,\|)$ and constants $0,1$ on event structures with disjoint sets of events. These definitions allow the inductive translation of program texts into event structure objects.

\begin{itemize}
\item[-] The algebraic constant $1$ is interpreted as $({e},\emptyset,\emptyset,\{(e,\delta)\},\{e\})$.
\item[-] The algebraic constant $0$ is interpreted as $(\emptyset,\emptyset,\emptyset,\emptyset,\emptyset)$.
\item[-] Each atomic action $r{\in}\Hip\Omega$ is associated with $(\{e\},\emptyset,\emptyset,\{(e,r)\},\{e\})$. This event structure is again denoted by $r$.
\item[-] The nondeterministic choice between the event structures $\EE$ and $\FF$ is constructed as 
\[
	\EE {+} \FF = (E{\cup} F,\#_{\EE{+}\FF},\mapsto_\EE{\cup}\mapsto_\FF,\lambda_\EE{\cup}\lambda_\FF,\{x{\cup} y\ |\ x{\in}\exit_\EE\wedge y{\in}\exit_\FF\})
\]
where $\#_{\EE{+}\FF} = [{\cup}_{x{\in}\exit_\EE\wedge y\exit_\FF}\textrm{sym}(x{\times} y)]{\cup} \#_\EE{\cup}\#_\FF{\cup}\textrm{sym}({\init(\EE){\times}\init(\FF)})$ and $\textrm{sym}$ is the symmetric closure of a relation on $E{\cup} F$. The square-bracketed set ensures that every final event in $\EE$ is in conflict with every final event in $\FF$. This ensures that, if $z{\in}\exit_{\EE{+}\FF}$, then $z\#z$.

\item[-] The sequential composition of $\EE$ by $\FF$ is 
\[
	\EE{\cdot}\FF = (E{\cup} F,\#_\EE{\cup}\#_\FF,\mapsto_\EE{\cup}\mapsto_\FF{\cup}\{x\mapsto e\ |\ e{\in}\init(\FF)\wedge x{\in}\exit_\EE \},\lambda_\EE{\cup}\lambda_\FF,\exit_\FF).
\]
\item[-] The concurrent composition of $\EE$ and $\FF$ is
\[
\EE\|\FF = (E{\cup} F,\#_\EE{\cup}\#_\FF,\mapsto_{\EE}{\cup} \mapsto_\FF,\lambda_\EE{\cup}\lambda_\FF,\exit_\EE{\cup}\exit_\FF).  \] 
\item[-] The binary Kleene star of $\EE$ and $\FF$ is the supremum of the sequence \[
	\FF, \FF {+} \EE{\seq}\FF, \FF {+} \EE{\seq}(\FF {+} \EE\seq\FF),\dots
\]
of bundle event structures with respect to the $\omega$-complete sub-BES order~\cite{Rab13b}.
\end{itemize}

\todo[inline]{\Tx The reviewer advised that this example should show the inductive construction in practice.}

\begin{example}
Let us consider the sequential programs $r,\delta{\in}\H_1\Omega$. A concurrent program that is skipping or running $r$ in parallel with itself is algebraically denoted by $(r\|r){+}1$. The construction of the associated event structure starts from the innermost operation $(r\|r)$, assuming that each occurrence of the atomic action $r$ is associated with an event from $\{e_r,e_r'\}$. Thus 
\[
	\EE_{r\|r} = (\{e_r,e_r'\}, \emptyset, \emptyset,\underbrace{\{(e_r,r),(e_r',r)\}}_{\lambda_{r\|r}}, \{\{e_r\},\{e_r'\}\}).
\]
We can now construct the nondeterministic choice between $r\|r$ and $\delta$ as
\begin{footnotesize}
\[
	\EE_{(r\|r){+}1} = (\{e_r,e_r',e_\delta\}, \{e_r\#e_\delta,e_r'\#e_\delta\},\emptyset,\lambda_{r\|r}{\cup}\{(e_\delta,\delta)\}, \{\{\epsilon,e_\delta\}\ |\ \epsilon{\in}\{e_r,e_r'\}\}).
\]
\end{footnotesize}
In this example, we have $e_r\#e_\delta$ and $e_r'\#e_\delta$ but $e_r$ and $e_r'$ are concurrent.
\end{example}

For every bundle event structure $\EE$, $0 {+} \EE {=} \EE$, $0{\cdot} \EE {=} \EE{\cdot} 0 {=} \EE$, and in particular, $0{\cdot} 1 {=} 1$. The constant $0$ was only introduced to have a bottom element on the set of bundle event structures with internal probabilities. It ensures that we can compute the Kleene star inductively from the least element. Moreover, $0$ will disappear in mixed expressions because of these properties.

We now show that the operations $(+)$ and $(\cdot)$ are preserved by the map $\sem{\ }$. The case of the binary Kleene star $({\bks})$ is proven in \Prop{pro:*-homomorphism}.

\begin{proposition}\label{pro:homomorphism}
For $\EE,\FF$ non-zero, feasible and terminating event structures, we have $\sem{\EE{+}\FF} = \sem{\EE}{+}\sem{\FF}$ and $\sem{\EE{\cdot}\FF} = \sem{\EE}{\cdot}\sem{\FF}$.
\end{proposition}


\begin{proof}
For the case of nondeterminism $(+)$, let $s{\in}\Omega$ be the initial state and $\mu{\in}\sem{\EE{+}\FF}(s)$. Let us firstly assume that $\mu {=} \sigma_s(\EE)$ for some $\sigma{\in}\sched_1(\EE{+}\FF)$. By definition of the sum $\EE{+}\FF$, the set of events $E$ and $F$ are disjoints, so we can define two schedulers $\sigma^\EE{\in}\sched_1(\EE)$ and $\sigma^\FF{\in}\sched_1(\FF)$ as follows. Let $\alpha{\in}\TT(\EE{+}\FF)$ and $(e,t){\in}\dom(\sigma(\alpha))$, we define

\begin{displaymath}
\sigma^\EE(\alpha)(e,t) = \begin{cases}
\sigma(\alpha)(e,t) & \textrm{if }\alpha{\in}\TT(\EE){\setminus}\{\emptyset\},\\
\frac{\sigma(\emptyset)(e,t)}{p^\EE_t} & \textrm{if } \alpha {=} \emptyset.
\end{cases}
\end{displaymath} 
where $p^\EE_t {=} \sum_{e'{\in}\init(\EE)}w(e,t)$, $w$ is the weight function associated to $\sigma$ at the trace $\emptyset$ and $s$ is the initial state. The real number $p^\EE_t$ is just a normalisation constant required by Property~\ref{pr:sched-prob} in the definition of schedulers.~\footnote{If $p^\EE_t {=} 0$, then $\sigma{\in}\sched_1(\FF)$.} The scheduler $\sigma^\FF$ is similarly defined. It follows directly from these definition of $\sigma^\EE$ and $\sigma^\FF$ that $\sigma(\emptyset)(e,t) {=} p^\EE_t\sigma(\emptyset)(e,t) {+} p^{\FF}_t\sigma(\emptyset)(e,t)$ where $p^\EE_t {+} p^\FF_t {=} 1$ because of Property~\ref{pr:sched-prob}. Hence, $\sigma_s(\EE) {=} p^\EE_s\sigma_s^\EE(\EE) {+} p^{\FF}_s\sigma^{\FF}_s(\FF)$ i.e. $\sigma_s(\EE){\in}\sem{\EE} {+} \sem{\FF}$. Since $\sem{\EE} {+} \sem{\FF}$ is convex and topologically closed, we deduce that $\sem{\EE {+} \FF}(s)\subseteq(\sem{\EE}{+}\sem{\FF})(s)$. 

For the converse inclusion $(\sem{\EE}{+}\sem{\FF})(s){\subseteq}\sem{\EE {+} \FF}(s)$, notice that $\overline{\conv(A)} = \conv(\overline{A})$ holds for every subset $A{\subseteq}\Real^\Omega$. If we write $A {=} \{\sigma_s(\EE) \ | \ \sigma{\in}\sched_1(\EE)\}$ and $B {=} \{\sigma_s(\FF) \ | \ \sigma{\in}\sched_1(\FF)\}$, then 
\[
	(\sem{\EE}{+}\sem{\FF})(s) = \overline{\conv(\overline{\conv(A)}{\cup}\overline{\conv(B)})} = \overline{\conv(A{\cup} B)}.
\]
But it is clear that $A{\subseteq}\sem{\EE{+}\FF}(s)$ (a scheduler that does not choose $\FF$ is possible because $\EE$ is feasible) and $B{\subseteq}\sem{\EE{+}\FF}(s)$. Therefore, $(\sem{\EE}{+}\sem{\FF})(s) = \overline{\conv(A{\cup} B)}{\subseteq} \sem{\EE{+}\FF}(s)$ because the last set is convex and topologically closed.

The sequential composition is proven using a similar reasoning. Let $\EE,\FF$ be two bundle event structures satisfying the hypothesis, and $\mu{\in}\sem{\EE{\cdot}\FF}(s)$ for some initial state $s{\in}\Omega$. 

The proof of $\sem{\EE{\cdot}\FF}(s)\subseteq \sem{\EE}{\cdot}\sem{\FF}(s)$ goes as follows. Firstly, let us assume that there is a scheduler $\sigma$ on $\EE{\cdot}\FF$ such that $\mu {=} \sigma_s(\EE{\cdot}\FF)$. Since schedulers are inductively constructed, there exists $\sigma^\EE{\in}\sched(\EE)$ and $\sigma^\FF{\in}\sched(\FF)$ such that 
\begin{displaymath}
\sigma(\alpha)(e,t) = 
\begin{cases}
\sigma^\EE(\alpha)(e,t) & \textrm{if }\alpha e{\in}\TT(\EE),\\
\sigma^\FF(\alpha'')(e,t) & \textrm{if } \alpha {=} \alpha'\alpha''\textrm{ and }(\alpha',\alpha''){\in} \TT_{\max}(\EE){\times}\TT(\FF).
\end{cases}
\end{displaymath}
Let us denote by $\varphi_n$ and $\varphi^\EE_n$ (resp. $\varphi^\FF_{n,t})\footnote{Remind that $\varphi_{n,t}$ is the computation function computed given the initial state $t$.}$ the computation sequences associated to the respective schedulers $\sigma$ and $\sigma^\EE$ (resp. $\sigma^\FF$) from the initial state $s$ (resp. $t$). It follows directly that $\varphi_n(\alpha) {=} \varphi_n^\EE(\alpha)$ for every $\alpha{\in}\TT_n(\EE)$. If $\alpha'{\in}\TT_{\max}(\EE){\cap}\TT_n(\EE)$ and $e{\in}\init(\FF)$ then, for every state $u{\in}\Omega$,
\[
	\varphi_{n{+}1}(\alpha' e)(u) = \sum_{t{\in}\Omega}\sigma^\FF(\emptyset)(e,t)(u)\varphi^\EE(\alpha')(t).
\]
Similarly, we have
\begin{align*}
\varphi_{n{+}1}(\alpha' ee')(u) & = \sum_{t'{\in}\Omega}\sigma^\FF(e)(e,t')(u)\left[\sum_{t{\in}\Omega}\sigma^\FF(\emptyset)(e,t)(t')\varphi^\EE(\alpha')(t)\right]\\
& = \sum_{t{\in}\Omega}\left[\sum_{t'{\in}\Omega}\sigma^\FF(e)(e,t')(u)\sigma^\FF(\emptyset)(e,t)(t')\right]\varphi^\EE(\alpha')(t) \\
& = \sum_{t{\in}\Omega} \varphi_{2,t}^\FF(u)\varphi^\EE(\alpha')(t).
\end{align*}

By simple induction on the length of $\alpha''$, we deduce that
\[
	\varphi(\alpha'\alpha'')(u) = \sum_{t{\in}\Omega}\varphi_{t}^\FF(\alpha'')(u)\varphi^\EE(\alpha')(t),
\]
where $\varphi^\FF_t$ is the complete run obtained from the sequence $\varphi^\FF_{n,t}$. 
It follows by definition of the sequential composition on $\H_1\Omega$ (\Eqn{eq:6-sequential-H}) that
\[
	\sigma_s(\EE)(u) = \sum_{t{\in}\Omega}\sigma^{\FF}_t(\FF)(u)\sigma^{\EE}_s(\EE)(t){\in} \sem{\EE}{\cdot}\sem{\FF}(s)
\]
for every state $u{\in}\Omega$. Secondly, since $\sem{\EE}{\cdot}\sem{\FF}(s)$ is upclosed and topologically closed, we deduce that $\sem{\EE{\cdot}\FF}(s)\subseteq \sem{\EE}{\cdot}\sem{\FF}(s)$. 

Conversely, if $\mu{\in} \sem{\EE}{\cdot}\sem{\FF}(s)$, then either $\mu(u) = \sum_{t{\in} \Omega}\sigma^{\FF}_t(\FF)(u)\sigma^{\EE}_s(\EE)(t)$ or $\mu$ is in the closure of the set of these distributions. Either way, the closure properties of $\sem{\EE{\cdot}\FF}(s)$ implies that $\sem{\EE}{\cdot}\sem{\FF}(s){\subseteq}\sem{\EE{\cdot}\FF}(s)$. 
\end{proof}

\subsection{Simulation for ipBES}\label{s1511}

The partial order defined in Definition~\ref{def:semantics-sequential} compares the sequential behaviours of two systems. However, it suffers from a congruence problem, i.e. there exist programs $\EE,\FF$ and $\GG$ such that $\EE{\refby}\FF$ but $\EE\|\GG {\not\refby}\FF\|\GG$. A known technique for achieving congruence is to construct an order based on simulations, which is the subject of this section. We use a similar technique in this subsection.

We say that a trace $\alpha$ is \emph{weakly maximal} if it is maximal or there exist some events $e_1,\dots,e_n$ such that $\alpha e_1\cdots e_n{\in}\TT_{\max}$ and $\unity{\refbyh}\lambda(e_i)$ for every $1\leq i\leq n$. 

\begin{definition}\label{def:t-simulation}
A function $f{:}\TT(\EE){\to}\TT(\FF)$ is called a \emph{t-simulation} if the following conditions hold:
\begin{itemize}
\item[-] if $f(\emptyset) = \emptyset$ and $f^{-1}(\beta)$ is a finite set for every $\beta{\in}\TT(\FF)$,
\item[-] if $\alpha e{\in}\TT(\EE)$ then either:
\begin{itemize}
\item $f(\alpha e) = f(\alpha)$ and $\lambda(e)\refbyh\delta$ holds in $\H_1\Omega$,
\item or there exists an event $e'$ of $\FF$ such that $\lambda(e){\refbyh}\lambda(e')$ and $f(\alpha e) {=} f(\alpha) e'$.
\end{itemize} 
\item[-] if $\alpha e$ is maximal in $\TT(\EE)$ then $f(\alpha e) = f(\alpha) e'$, for some $e'$ (with $\lambda(e)\refbyh\lambda(e')$), and $f(\alpha e)$ is weakly maximal in $\TT(\FF)$~\footnote{If $f(\alpha)$ is maximal then $\alpha$ is necessarily maximal.}.
\end{itemize}
We say that $\EE$ is simulated by $\FF$, written $\EE\simref\FF$, if there exists a simulation from $\EE$ to $\FF$. The equivalence generated by this preorder is denoted $\equiv_{\textrm{sim}}$.
\end{definition}

The notion of t-simulation has been designed to simulate event structures correctly in the presence of tests. For instance, given a test $b$, the simulation $\unity\simref(b{+}\neg b)$ fails because a t-simulation is a total function and it does not allow the removal of ``internal" events labelled with subidentities during a refinement step. The finiteness condition on $f^{-1}(\beta)$ ensures that we do not refine a terminating specification with a diverging implementation. Without that constraint, we would be able to write the refinement \[
	\mathtt{if}\ (0{=}1)\ \mathtt{then}\ s{:=}0\ \mathtt{ else } [\mathtt{if}\ (0{=}1)\ \mathtt{then}\ s{:=}0\ \mathtt{else} [\dots]]\simref s{:=}0.
\]
However, this should not hold because the left hand sides is a non-terminating program and cannot refine the terminating assignment ${s:=}0$.

A t-simulation is used to compare bundle event structures without looking in details at the labels of events. It can be seen as a refinement order on the higher level structure of a concurrent program.  Once a sequential behaviour has to be checked, we use the previously defined functional equivalence on event structures with internal probabilities.

\begin{example}\label{ex:t-sim}
Consider a program variable $x$ of type Boolean (with value $0$ or $1$). A t-simulation from $(x{=}1){+}(x{\neq} 1){\cdot} (x{:=}1)$ to $1 {+} (x{:=}0{\sqcap} 1)$~\footnote{$x{:=}0{\sqcap} 1$ is an atomic nondeterministic-assignment, it cannot be interfered with.} is given by the dotted arrow in the following diagram:
\begin{displaymath}
\xymatrix{
&\emptyset\ar[dl]\ar[dr]\ar@{.>}[rrrr]& & && \ar[dl]\ar[dr]\emptyset&\\
e_{x{=}1}\ar@{.>}@/_/[rrrr]& & e_{x{\neq}1}\ar[d]\ar@{.>}[urrr] &&e_\delta&&e_{x{:=}0{\sqcap} 1} \\
&&e_{x{\neq}1}e_{x{:=}1}\ar@{.>}[urrrr]&&&&
}
\end{displaymath}
This t-simulation refines two nondeterministic choices, one at the program structure level and the other at the atomic level.
\end{example}

\begin{proposition}
The t-simulation relation $\simref$ is a preorder.
\end{proposition}

\begin{proof}
Reflexivity follows from the identity function and transitivity is obtained by composing t-simulations which will generate a new t-simulation. Notice that care should be taken with respect to the third property of a t-simulation. If $f{:}\TT(\EE){\to}\TT(\FF)$, $g{:}\TT(\FF){\to}\TT(\GG)$ are t-simulations, $\alpha e{\in}\TT_{\max}(\EE)$ and $\lambda(e){\refbyh}\unity$, then $f(\alpha e) {=} f(\alpha) e'$ for some $e'$ of $\FF$ such that $\lambda(e){\refbyh}\lambda(e')$. If $\lambda(e'){\refbyh}\unity$, then it is possible that $g(f(\alpha)e') {=} g(f(\alpha))$. However, since $f(\alpha)e'$ is weakly maximal, $g(f(\alpha)e')$ is also weakly maximal and we can find an event $e''{\in} G$ such that $f(\alpha)e''$ is weakly maximal and $\lambda(e'){\refbyh}\lambda(e'')$. We then map $\alpha e$ to $g(f(\alpha))e''$ in the t-simulation from $\EE$ to $\FF$.
\end{proof}

\begin{proposition}\label{pro:necessary-axioms}
If $\EE,\FF,\GG$ are ipBES, then
\begin{align}
\EE \| \FF &\equiv_{\mathrm{sim}} \FF\|\EE,\label{eq:par-comm}\\
\EE \| (\FF \| \GG) &\equiv_{\mathrm{sim}} (\EE\|\FF)\|\GG,\label{eq:par-assoc}\\
\EE{\bks} \FF&\equiv_{\mathrm{sim}}\FF {+} \EE{\cdot} (\EE{\bks} \FF),\label{eq:unfold}\\
\EE\simref\FF & \Rightarrow\GG{+}\EE\simref\GG{+}\FF,\label{eq:+-monotony}\\
\EE\simref\FF & \Rightarrow\GG{\cdot}\EE\simref\GG{\cdot}\FF,\label{eq:monotony}\\
\EE\simref\FF & \Rightarrow \EE\|\GG\simref\FF\|\GG.\label{eq:congruence}
\end{align}
\end{proposition}

\begin{proofsummary}
The constructions $\EE\|\FF$ and $\FF\|\EE$ result in the same event structure and similarly for associativity. 

For the implication~(\ref{eq:congruence}), let $f{:}\TT(\EE){\to}\TT(\FF)$ be a t-simulation. We construct a t-simulation $g{:}\TT(\EE\|\GG){\to}\TT(\FF\|\GG)$ inductively. We set $g(\emptyset) {=} \emptyset$. Let $\alpha{\in} \TT(\EE\|\GG)$ and $e{\in} E{\cup} G$ such that $\alpha e$ is a trace of $\EE\|\GG$. We write $\alpha|_E$ for the restriction of $\alpha$ to the events occurring in $\EE$. The inductive definition of $g$ is:
\begin{displaymath}
g(\alpha e) = \begin{cases}
g(\alpha)e & \textrm{if } e{\in} G,\\
g(\alpha) & \textrm {if } e{\in} E\textrm{ and } f(\alpha|_Ee) {=} f(\alpha|_E), \\
g(\alpha)e' & \textrm{if } e{\in} E\textrm{ and } f(\alpha|_Ee) {=} f(\alpha|_E)e'
\end{cases}
\end{displaymath}
Since the set of events of $\EE$ and $\GG$ are disjoint, the cases in the above definition of $g$ are disjoint. That is, $g$ is indeed a function and it satisfies the second property of a t-simulation. The last property is clear because if $\alpha e$ is maximal in $\TT(\EE\|\GG)$, then either $\alpha|_E$ is maximal in $\EE$ and $\alpha|_Ge$ is maximal in $\TT(\GG)$, or $\alpha|_Ee$ is maximal in $\TT(\EE)$ and $\alpha|_G$ is maximal in $\TT(\GG)$. In both cases, $g(\alpha e) {=} g(\alpha)e'$ for some $e'{\in} E{\cup} G$ and $g(\alpha e)$ is weakly maximal in $\TT(\FF\|\GG)$.

Similarly, the other cases are shown by constructing t-simulations.
\end{proofsummary}

\begin{proof}
The constructions $\EE\|\FF$ and $\FF\|\EE$ result in the same event structure and similarly for the associativity. 


The Unfold \Eqn{eq:unfold} is clear because the left and right hand side event structures are exactly the same up to renaming of events.

Implication~(\ref{eq:+-monotony}) follows by considering the function $\mathrm{id}_{\TT(\GG)}{\cup} f{:}\TT(\GG{+}\EE){\to}\TT(\GG {+} \FF)$. It is indeed a function because the sets of events $G$ and $E$ (resp. $F$) are disjoint. The property of a t-simulation follows directly because the set of traces $\TT(\GG{+}\EE)$ is the disjoint union $\TT(\GG){\cup}\TT(\EE)$ (similarly for $\GG{+}\FF$).

For case of sequential composition~(\ref{eq:monotony}), let $f$ be a t-simulation from $\EE$ to $\FF$. It is clear that the function $g{:}\TT(\GG{\cdot}\EE){\to}\TT(\GG{\cdot}\FF)$, such that $g(\alpha) {=} \alpha|_Gf(\alpha|_E)$ is a t-simulation.

For the Implication~(\ref{eq:congruence}), let $f{:}\TT(\EE){\to}\TT(\FF)$ be a t-simulation. Let us construct a t-simulation $g{:}\TT(\EE\|\GG){\to}\TT(\FF\|\GG)$ inductively. We set $g(\emptyset) {=} \emptyset$. Let $\alpha{\in} \TT(\EE\|\GG)$ and $e{\in} E{\cup} G$ such that $\alpha e$ is a trace of $\EE\|\GG$. We write $\alpha|_E$ the restriction of $\alpha$ to the events occurring in $\EE$. The inductive definition of $g$ is:
\begin{displaymath}
g(\alpha e) = \left\lbrace
\begin{array}{cl}
g(\alpha)e & \textrm{if } e{\in} G,\\
g(\alpha) & \textrm {if } e{\in} E\textrm{ and } f(\alpha|_Ee) {=} f(\alpha|_E), \\
g(\alpha)e' & \textrm{if } e{\in} E\textrm{ and } f(\alpha|_Ee) {=} f(\alpha|_E)e'.
\end{array}\right.
\end{displaymath}
Since the set of events of $\EE$ and $\GG$ are disjoint, the cases in the above definition of $g$ are disjoint. That is, $g$ is indeed a function and it satisfies the second property of a t-simulation. The last property is clear because if $\alpha e$ is maximal in $\TT(\EE\|\GG)$, then either $\alpha|_E$ is maximal in $\EE$ and $\alpha|_Ge$ is maximal in $\TT(\GG)$, or $\alpha|_Ee$ is maximal in $\TT(\EE)$ and $\alpha|_G$ is maximal in $\TT(\GG)$. In both cases, $g(\alpha e) {=} g(\alpha)e'$ for some $e'{\in} E{\cup} G$ and $g(\alpha e)$ is weakly maximal in $\TT(\FF\|\GG)$.
\end{proof}

We now state the main result of this section, which is the backbone of our probabilistic rely-guarantee calculus.

\begin{theorem}\label{thm:trace-imply-distribution}
Let $\EE$ and $\FF$ be feasible and terminating ipBES. Then $\EE\simref\FF$ implies $\EE\refby\FF$.
\end{theorem}

\begin{proofsummary}
The proof amounts to showing that, given an initial state $s{\in}\Omega$ and a scheduler $\sigma$ of $\EE$, there exists a scheduler $\sigma'$ of $\FF$ that generates exactly the same distribution as $\sigma$ from the state $s$. The scheduler $\sigma'$ is  constructed  inductively from the t-simulation from $\EE$ to $\FF$. 
\end{proofsummary}

\begin{proof}
Let $f$ be a t-simulation from $\EE$ to $\FF$, $s{\in}\Omega$ be the initial state, $\sigma{\in}\sched_1(\EE)$ and $\varphi$ is the complete run of $\sigma$ on $\EE$ from $s$. We have to generate a scheduler $\tau{\in}\sched_1(\FF)$ such that the measures $\sigma_s(\EE)$ and  $\tau_s(\FF)$ are equal i.e. they produce the same value for every state $u{\in}\Omega$.

For every $\beta{\in}\TT(\FF)$, we define $f^{-1}_{\min}(\beta)$ to be the set of minimal traces in $f^{-1}(\beta)$, that is,
\[
	f_{\min}^{-1}(\beta) = \{\alpha\ |\ \forall e{\in} E:\alpha {=} \alpha' e{\in} f^{-1}(\beta)\Rightarrow \alpha' {\notin} f^{-1}(\beta)\}.
\]
We now construct the scheduler $\tau$. Let $\beta{\in}\TT(\FF)$. We consider two cases:
\begin{itemize}
\item If $f^{-1}(\beta) = \emptyset$ then we set $\tau(\beta)(e,t) = 0{\in}\Jip\Omega$, except for some particular maximal traces that are handled in $(\dagger)$ below. 
\item Otherwise, given a state $t{\in}\Omega$, we define a normalisation factor 
\[
	C_{\beta,t} = \sum_{\alpha{\in} f_{\min}^{-1}(\beta)}{\varphi}(\alpha)(t),
\]
and we set~\footnote{Notice if $C_{\beta,t} = 0$ for some $t{\in}\Omega$ then ${\varphi}(\alpha)(t) = 0$ for every $\alpha{\in} f_{\min}^{-1}(\beta)$. In other words, none of these $\alpha$ will be scheduled at all. Hence, $\beta$ need not be scheduled either.}
\[
	\tau(\beta)(e,t) = \frac{1}{C_{\beta,t}}\left(\sum_{\alpha{\in} f_{\min}^{-1}(\beta)}\varphi(\alpha)(t)\sum_{\alpha e_1\cdots e_k{\in} f_{\min}^{-1}(\beta e)}\prod_{i=1}^{k}w_{i-1}(e_i,t)\mu_k\right)
\]
where $w_{i-1}(e_i,t)$ is the weight function such that $\sigma(\alpha e_1\cdots e_{i-1})(e_i,t) = w_{i-1}(e_i)\mu$, and $\mu{\in}\lambda(e_{i})$ (if $\lambda(e_{i})(t)$ is empty then $w_{i-1}(e_i,t) = 0$). The distribution $\mu_k$ is chosen by $\sigma$ from $\lambda(e_k)(t)$, when scheduling $e_k$.
\end{itemize}

Firstly, we show that $\tau$ is indeed a scheduler on $\FF$. The Property(\ref{pr:sched-dom}) of Definition~\ref{def:ipscheduler} is clear. Let us show the other properties. Let $\beta e{\in}\TT(\EE)$ and
let $W{:}E{\times}\Omega{\to}\Real$ be the weight function such that 
\[
	W(e,t) = \frac{1}{C_{\beta,t}}\sum_{\alpha{\in} f_{\min}^{-1}(\beta)}\varphi(\alpha)(t)\sum_{\alpha e_1\cdots e_k{\in} f_{\min}^{-1}(\beta e)}\prod_{i=1}^{k}w_{i-1}(e_i,t).
\]
Indeed, $\mu {=} \frac{\tau(\beta)(e,t)}{W(e,t)}$ is in $\lambda(e)(t)$~\footnote{The case $W(e,t) = 0$ can be adapted easily because the numerator in the definition of $\tau(\beta)(e)$ is also $0$. For instance, we can assume that $\frac{0}{0} = 1$.} because $\lambda(e)(t)$ is convex and for each $\alpha e_1\cdots e_k{\in}f_{\min}^{-1}(\beta e)$ and $\mu_k{\in}\lambda(e_k)(t){\subseteq}\lambda(e)(t)$.
Hence $\tau(\beta)(e) = W_ef_e$ and $\tau$ satisfies the Property (\ref{pr:sched-choice}) of \Defs{\ref{def:ipscheduler}}. As for Property (\ref{pr:sched-prob}), let $s{\in}\Omega$ and let us compute the quantity 
\[
	V(t) = \sum_{(e,t){\in}\dom(\tau(\beta))}W(e,t),
\]
for a fixed $t{\in}\Omega$. Let us write $\dom(\beta) = \{e \ | \ \beta e{\in}\TT(\FF) \}$.
\begin{eqnarray}
V(t) & = & \sum_{(e,t){\in}\dom(\tau(\beta))}\frac{1}{C_{\beta,t}}\sum_{\alpha{\in} f_{\min}^{-1}(\beta)}\varphi(\alpha)(t)\sum_{\alpha e_1\cdots e_k{\in} f_{\min}^{-1}(\beta e)}\prod_{i=1}^{k}w_{i-1}(e_i,t)\nonumber\\
& = & \frac{1}{C_{\beta,t}}\sum_{\alpha{\in} f_{\min}^{-1}(\beta)}\varphi(\alpha)(t)\sum_{(e,t){\in}\dom(\tau(\beta))}\sum_{\alpha e_1\cdots e_k{\in} f_{\min}^{-1}(\beta e)}\prod_{i=1}^{k}w_{i-1}(e_i,t)\nonumber\\
 & = & \frac{1}{C_{\beta,t}}\sum_{\alpha{\in} f_{\min}^{-1}(\beta)}\varphi(\alpha)(t)\sum_{\alpha e_1\cdots e_k{\in} {\cup}_{e{\in}\dom(\beta)}f_{\min}^{-1}(\beta e)}\prod_{i=1}^{k}w_{i-1}(e_i,t).\nonumber
\end{eqnarray}
From the second to the third expression, the two rightmost sums were merged into a single one because $f^{-1}_{\min}(\beta e){\cap} f^{-1}(\beta e') = \emptyset$ ($f$ is a function). It follows from Property~(\ref{pr:sched-prob}), applied on the weight $w_{i-1}(e_i,t)$ of $\sigma$, that
\[
	\sum_{\alpha e_1\cdots e_k{\in} {\cup}_{e{\in}\dom(\beta)}f_{\min}^{-1}(\beta e)}\prod_{i=1}^{k}w_{i-1}(e_i,s) = 1
\]	
and hence $V = 1$ (c.f. Figure~\ref{fig:6-concrete-sched} for a concrete example). The last Property (\ref{pr:sched-consistent}) of \Defs{\ref{def:ipscheduler}} is clear because if $\lambda(e)(t) = \emptyset$, then the coefficient of $\sigma(\alpha e_1\cdots e_{k-1})(e_k,t)$ is $0$ because $\lambda(e_k)(t) = \emptyset$. Hence, the product is also $0$.

\begin{figure}
\begin{displaymath}
\xymatrix{
&\alpha\ar[dl]_{w_0(e'_1,t)}\ar[dr]^{w_0(e_1,t)}\ar@{.>}[rrrr]&&&&\beta\ar[d]^{W(e,t)}&\\
\alpha e_1'\ar@/_/@{.>}[rrrrr]&&\alpha e_1\ar[dl]_{w_1(e'_2,t)}\ar[dr]^{w_1(e_2,t)}\ar@/_/@{.>}[rurr]&&&e&\\
&\alpha e_1e_2'\ar@/_/@{.>}[urrrr]&&\alpha e_1e_2\ar@/_/@{.>}[urr]&&&
}
\end{displaymath}
We have $V(t) =w_0(e'_1,t) + w_0(e_1,t)w_1(e_2',t)w_0(e_1,t)w_1(e_2,t) = 1$ because $w_1(e_2',t)+w_1(e_2,t) = 1$ and $w_0(e'_1,t) + w_0(e_1,t)=1$ (\Defs{\ref{def:ipscheduler}} Property (\ref{pr:sched-prob})).
\caption{An example showing that $V(t) = 1$.}\label{fig:6-concrete-sched}
\end{figure}

Secondly, let $\psi$ be the complete run of $\FF$ with respect to $\tau$. We now show by induction on $\beta$ that 
\begin{equation}\label{eq:subgoal}
\psi(\beta) = \sum_{\alpha{\in} f_{\min}^{-1}(\beta)}\varphi(\alpha) = C_{\beta,t},
\end{equation}
where the empty sum evaluates to the identically zero distribution. The base case is clear because $\psi(\emptyset) {=} \delta_{s} {=} \phi(\emptyset)$ where $s$ is the initial state. Let us assume the above identity for $\beta{\in}\TT(\FF)$ and let $e{\in} F$ such that $\beta e {=} \TT(\EE)$ and $f_{\min}^{-1}(\beta e){\neq}\emptyset$. By definition of $\psi$, if $u{\in}\Omega$, we have:
\begin{footnotesize}
\begin{align*}
\psi(\beta e)(u) &= \sum_{t{\in}\Omega} \frac{1}{C_{\beta,t}}\sum_{\alpha{\in} f_{\min}^{-1}(\beta)}\varphi(\alpha)(t)\sum_{\alpha e_1\cdots e_k{\in} f_{\min}^{-1}(\beta e)}\prod_{i=1}^{k}w_{i-1}(e_i,t)\mu_k(u)\psi(\beta)(t)\\
& = \sum_{t{\in}\Omega}\sum_{\alpha{\in} f_{\min}^{-1}(\beta)}\sum_{\alpha e_1\cdots e_k{\in} f_{\min}^{-1}(\beta e)}\prod_{i=1}^{k}w_{i-1}(e_i,t)\mu_k(u)\varphi(\alpha)(t) \\
& = \sum_{\alpha{\in} f_{\min}^{-1}(\beta)}\sum_{\alpha e_1\cdots e_k{\in} f_{\min}^{-1}(\beta e)}\sum_{t{\in} \Omega}\prod_{i=1}^{k}w_{i-1}(e_i,t)\mu_k(u)\varphi(\alpha)(t)\\
& =  \sum_{\alpha{\in} f_{\min}^{-1}(\beta)}\sum_{\alpha e_1\cdots e_k{\in} f_{\min}^{-1}(\beta e)}\sum_{t{\in} \Omega} \sum_{t'{\in}\Omega} w_0(e_1,t')\delta_{t'}(t)\left[\prod_{i=2}^{k}w_{i-1}(e_i,t)\mu_k(u)\right]\varphi(\alpha)(t')\\
& =\sum_{\alpha{\in} f_{\min}^{-1}(\beta)}\sum_{\alpha e_1\cdots e_k{\in} f_{\min}^{-1}(\beta e)}\sum_{t{\in} \Omega}\prod_{i=2}^{k}w_{i-1}(e_i,t)\mu_k(u)\varphi(\alpha e_1)(t)\\
& =  \sum_{\alpha{\in} f_{\min}^{-1}(\beta)}\sum_{\alpha e_1\cdots e_k{\in} f_{\min}^{-1}(\beta e)}\sum_{t{\in} \Omega} \sum_{t'{\in}\Omega} w_1(e_2,t')\delta_{t'}(t)\left[\prod_{i=3}^{k}w_{i-1}(e_i,t)\mu_k(u)\right]\varphi(\alpha e_1)(t')\\
& = \cdots .
\end{align*}
\end{footnotesize}
By continuing the above reasoning for all $e_i$ (induction), $i\leq k-1$, we obtain 
\begin{align*}
\psi(\beta e)(u) & = \sum_{\alpha{\in} f_{\min}^{-1}(\beta)}\sum_{\alpha e_1\cdots e_k{\in} f_{\min}^{-1}(\beta e)}\sum_{t{\in} \Omega}w_{k-1}(e_k,t)\mu_k(u)\varphi(\alpha e_1\cdots e_{k-1})(t)\nonumber\\
& = \sum_{\alpha{\in} f_{\min}^{-1}(\beta)}\sum_{\alpha e_1\cdots e_k{\in} f_{\min}^{-1}(\beta e)}\varphi(\alpha e_1\cdots e_k)(u).\nonumber
\end{align*}
Hence, 
\[
	\psi(\beta e)(u) = \sum_{\alpha'{\in} f_{\min}^{-1}(\beta e)}\varphi(\alpha')(u).
\]
$(\dagger)$ We finally compute the sum $\tau_s(\FF) = \sum_{\beta{\in}\TT_{\max}(\FF)}\psi(\beta)$. Notice firstly that $\tau$ may not schedule some traces of $\FF$. In particular, the third property in the definition of simulation implies that a maximal element of $\TT(\EE)$ may be mapped to a weakly maximal element of $\TT(\FF)$. Hence, we need to extend the scheduler $\tau$ so that it is non-zero for exactly one maximal element from that weakly maximal trace. More precisely, if $\beta' = f(\alpha)$ is weakly maximal for some maximal trace $\alpha{\in}\TT_{\max}(\EE)$, then there exists a sequence $e_1,\dots,e_n$ such that $\beta = \beta' e_1\cdots e_n{\in}\TT_{\max}(\FF)$ and $\unity\refbyh\lambda(e_i)$. We extend $\tau$ such that $\tau(\beta' e_1\cdots e_i)(e_{i{+}1},t) = \delta_t$. This implies that $\psi(\beta)(t) = \psi(\beta')(t)$. The other case is that $\beta$ is maximal and belongs to the image of $f$. In both cases, we have 
\[
	\psi(\beta)(t) = \sum_{\alpha{\in} A_\beta}\varphi(\alpha)(t),
\]
where $A_\beta = f^{-1}_{\min}(\beta)$ if $\beta$ is in the image of $f$, or $A_\beta = f^{-1}_{\min}(\beta')$ if there is such a $\beta'$ as above, otherwise, $A_\beta = \emptyset$. Thus, $A_\beta$ contains maximal traces only (if it is not empty). Since, $f$ is a total function, the set $\{A_\beta\ |\ \beta{\in}\TT_{\max}(\FF)\}$ is a partition of $\TT_{\max}(\EE)$ and we have 
\[
	\sum_{\beta{\in}\TT_{\max}(\EE)}\psi(\beta)(t) = \sum_{\beta{\in}\TT_{\max}(\EE)}\sum_{\alpha{\in} A_\beta}\varphi(\alpha)(t) = \sum_{\alpha{\in}\TT_{\max}(\EE)}\varphi(\alpha)(t),
\]
i.e. we obtain $\tau_s(\FF) = \sigma_s(\EE)$.
\end{proof}

\begin{example}
Reconsider the t-simulation of Example~\ref{ex:t-sim}. By definition, the unique scheduler $\sigma$ on $(x{=}1) {+} (x{\neq}1){\cdot} (x{:=}1)$ satisfies:
\begin{itemize}
\item[-] $\sigma(\emptyset)(e_{x{=}t},s) = w(e_{x{=}t},s)\delta_t$ where
$w(e_{x{=}t},s) = \begin{cases}
1 & \textrm{if } s {=} t, \\
0 & \textrm{otherwise.}
\end{cases}$
\item[-] $\sigma(e_{x{=}0})(e_{x{:=}1},s) = \delta_1$, for $s{\in}\{0,1\}$.
\end{itemize}
The corresponding scheduler $\tau$ on $1{+}(x{:=}0{\sqcap}1)$, constructed (as per the proof of \Thm{thm:trace-imply-distribution}) from $\sigma$ using the illustrated t-simulation, satisfies:
\begin{itemize}
\item[-] $\sigma(\emptyset)(e_\delta,1) = w(e_{x{=}1},1)\delta_1 = \delta_1$ and $\sigma(\emptyset)(e_\delta,0) = 0$,
\item[-] $\sigma(\emptyset)(e_{x:=0{\sqcap}1},1)=0$ and $\sigma(\emptyset)(e_{x{:=}0{\sqcap}1},0) = w(e_{x{=}0},0)\delta_1 = \delta_1$.
\end{itemize}
Since $(x{=}1) {+} (x{\neq}1){\cdot} (x{:=}1)$ is sequentially equivalent to $x{:=}1$, we can see that the scheduler $\tau$ on $1{+}(x{:=}0{\sqcap}1)$ forces the final value of $x$ to be $1$ by resolving $(+)$ and $(\sqcap)$ as they were resolved in the program $(x{=}1) {+} (x{\neq}1){\cdot} (x{:=}1)$.
\end{example}

We now show that the binary Kleene star is preserved by the semantics map.

\begin{proposition}\label{pro:*-homomorphism}
For every non-zero, feasible and terminating event structure $\EE$ and $\FF$, we have $\sem{\EE{\bks}\FF} = \sem{\EE}{\bks}\sem{\FF}$.
\end{proposition}
\begin{proofsummary}
  Note that $\sem{\EE}{\bks}\sem{\FF}$ is the least fixed point of the function $f(X) = \sem{\FF}{+} \sem{\EE}{\cdot} X$ in $\H_1\Omega$\footnote{Notice that the least fixed point is in $\H_1\Omega$ but not $\Hip\Omega$. The reason is that $\sem{\EE}$ and $\sem{\FF}$ are elements of $\H_1\Omega$ because of feasibility and termination.}, and $\EE{\bks}\FF$ satisfies $\FF {+} \EE{\cdot}(\EE{\bks}\FF) \equiv_{\textrm{sim}} \EE{\bks}\FF$ by construction of the sequences of bundle event structures defining $\EE{\bks}\FF$. Therefore, \Thm{thm:trace-imply-distribution} and \Prop{pro:homomorphism} imply that $\ \sem{\EE}{\bks}\sem{\FF}\refbyh \sem{\EE{\bks}\FF}$.

The converse refinement $\ \sem{\EE{\bks}\FF}\refbyh\sem{\EE}{\bks}\sem{\FF}$ holds because every scheduler of $\EE{\bks}\FF$ is the ``limit" of a sequence of schedulers used in the construction of $\sem{\EE}{\bks}\sem{\FF}$.
\end{proofsummary}

\begin{proof}
For the binary Kleene product, since $\sem{\EE}{\bks}\sem{\FF}$ is the least fixed point of $f(X) = \sem{\FF}{+} \sem{\EE}{\cdot} X$ in $\H_1\Omega$\footnote{Notice that the least fixed point is in $\H_1\Omega$ but not $\Hip\Omega$. The reason is that $\sem{\EE}$ and $\sem{\FF}$ are elements of $\H_1\Omega$ because of feasibility and termination.}, and $\EE{\bks}\FF$ satisfies 
\[
	\FF {+} \EE{\cdot}(\EE{\bks}\FF) \equiv_{\textrm{sim}} \EE{\bks}\FF
\]
by construction of the sequences of bundle event structures defining $\EE{\bks}\FF$. Therefore, \Thm{thm:trace-imply-distribution} and \Prop{pro:homomorphism} imply that $\ \sem{\EE}{\bks}\sem{\FF}\refbyh \sem{\EE{\bks}\FF}$. 

Conversely, let $\mu{\in} \sem{\EE{\bks}\FF}(s)$ for some initial state $s{\in}\Omega$. As in the case of \Prop{pro:homomorphism}, we assume that $\mu$ is computed from a scheduler $\sigma$ on $\EE{\bks}\FF$. We construct a sequence of schedulers $\sigma_n$ that ``converges" to $\sigma$ as follows. We set $\sigma_0$ to be any element of $\sched_1(\FF)$, $\sigma_1(\alpha) {=} \sigma(\alpha)$ if $\alpha$ is a trace of $\FF$ or $\EE$, otherwise, we set $\sigma_1(\alpha'\alpha'') {=} \sigma_0(\alpha'')$ where $\alpha'{\in}\TT_{\max}(\EE)$ (notice that $\sigma_0$ is applied to a different copy of $\FF$ but this is not important as event names can be abstracted.). Inductively, we define 
\begin{displaymath}
\sigma_{n}(\alpha) =\begin{cases}
\sigma(\alpha)&\textrm{if } \alpha{\in} \TT(\underbrace{\FF {+} \EE{\cdot}(\dots \EE{\cdot}(\FF {+} \EE))}_{n \textrm{ occurrences of }\EE}), \\
\sigma_{0}(\alpha|_F)&\textrm{otherwise.}
\end{cases}
\end{displaymath}  

Again, $\sigma_0$ is applied to the $n{+}1^{\textrm{th}}$ copy of $\FF$. Indeed, we have 
\[
	\sigma_n{\in}\sched_1(\underbrace{\FF {+} \EE{\cdot}(\cdots\EE{\cdot}(\FF {+} \EE{\cdot}\FF))}_{n\textrm{ occurrences of } \EE})
\] 
by construction. On the one hand, the sequence of distributions $\sigma_{n,s}(\EE)$ forms a subset of $\sem{\EE}{\bks}\sem{\FF}(s)$. On the other hand, let $u{\in}\Omega$ and let us denote 
\[
	\TT_{\leq n} = \TT(\underbrace{\FF {+} \EE{\cdot}(\dots \EE{\cdot}(\FF {+} \EE{\cdot}\FF))}_{n \textrm{ occurrences of }\EE}).
\]	
If we denote by $\varphi_n$ the complete run of $\sigma_n$ on $\EE{\bks}\FF$, then we have
\begin{eqnarray*}
\left|\sigma_s(\EE)(u) - \sigma_{n,s}(\alpha)(u)\right| & = & \left|\sum_{\alpha{\in}\TT_{\max}(\EE{\bks}\FF)}\varphi(\alpha)(u){-}\sum_{\alpha{\in}\TT_n{\cap}\TT_{\max}(\EE{\bks}\FF)}\varphi_n(\alpha)(u)\right|\\
& = & \left|\sum_{\alpha{\in}\TT_{\max}(\EE{\bks}\FF){\setminus}\TT_{\leq n{-}1}}\left(\varphi(\alpha)(u) {-} \varphi_n(\alpha)(u)\right)\right|\\
&\leq& \sum_{\alpha{\in}\TT_{\max}(\EE{\bks}\FF){\setminus}\TT_{\leq n{-}1}}\left|\varphi(\alpha)(u) {-} \varphi_n(\alpha)(u)\right|.\\
\end{eqnarray*}
The set $\TT_{\max}(\EE{\bks}\FF){\setminus}\TT_{\leq n{-}1}$ shrinks, when $n$ increases, because every finite trace of $\EE{\bks}\FF$ belongs to some set $\TT_{\leq k}$. Therefore, the last sum above is decreasing to $0$. Hence, since $\Omega$ is a finite set, the sequence $\sigma_{n,s}(\EE{\bks}\FF)$ converges (pointwise) to $\sigma_s(\EE{\bks}\FF)$ in $\D\Omega$. Since $\sem{\EE}{\bks}\sem{\FF}(s)$ is topologically closed, we deduce that $\sigma_s(\EE){\in}\sem{\EE}{\bks}\sem{\FF}(s)$. Therefore, $\sem{\EE{\bks}\FF}\refbyh\sem{\EE}{\bks}\sem{\FF}$.
\end{proof}

\begin{proposition}\label{pro:get-rid-of-par}
Let $r,r'{\in}\H_1\Omega$ be two atomic programs and let $\EE,\FF$ be two bundle event structures with internal probability, then 
\begin{align}
r^{\bks}\|r^{\bks}&\simref r^{\bks},\label{eq:star-trans}\\
r^{\bks}\|r'&\simref r{\bks}(r'{\cdot} r^{\bks}),\label{eq:par-atom}\\
r^{\bks}\|(b{\cdot} \EE {+} c{\cdot}\FF)& \simref  r{\bks}(b{\cdot} (r^{\bks}\|\EE)
{+}c{\cdot} (r^{\bks}\|\EE)),\label{eq:par-nondet}\\
r^{\bks}\|(r'{\cdot}\EE)&\simref r{\bks}(r'{\cdot} (r^{\bks}\|\EE)), \label{eq:par-seq}
\end{align}
where $r^{\bks} = r{\bks} 1$.
\end{proposition}

\begin{proofsummary}
  These inequations are again verified by explicitly constructing a t-simulation from the left-hand side to the right-hand side of the inequality. The following reasoning illustrates such a construction for \Eqn{eq:star-trans}.

  Denote by $e_1$ and $e_2$ (resp. $e$) the events that are labelled by $\delta$ in the event structure associated to $r^{\bks}\|r^{\bks}$ (resp. $r^{\bks}$). We construct an operation $(')$ such that, given a trace $\alpha$ of $r^{\bks}\|r^{\bks}$ that does not contain any of the $e_i$s, we define $\alpha'$ to be the unique trace corresponding to $\alpha$ in $r^{\bks}$ (i.e. with the same number of events labelled by $r$). A t-simulation from $r^{\bks}\|r^{\bks}$ to $r^{\bks}$ is obtained by considering a function $f$ such that \begin{displaymath} f(\alpha) = \begin{cases}
      (\alpha{\setminus}\{e_1,e_2\})' & \textrm{if } e_1{\notin}\alpha\textrm{ or }e_2{\notin}\alpha,\\
      (\alpha{\setminus}\{e_1,e_2\})'e & \textrm{if } e_1,e_2{\in}\alpha.\end{cases} \end{displaymath} \end{proofsummary}

\begin{proof}
Let us denote by $e_1$ and $e_2$ (resp. $e$) the events that are labelled by $\delta$ in the event structure associated to $r^{\bks}\|r^{\bks}$ (resp. $r^{\bks}$). Given a trace $\alpha$ of $r^{\bks}\|r^{\bks}$ that does not contain any of the $e_i$s, we denote by $\alpha'$ unique trace corresponding to $\alpha$ in $r^{\bks}$ (i.e. with the same number of events labelled by $r$). 

A t-simulation from $r^{\bks}\|r^{\bks}$ to $r^{\bks}$ is obtained by considering a function $f$ such that 
\begin{displaymath}
f(\alpha) = \left\lbrace
\begin{array}{cl}
(\alpha{\setminus}\{e_1,e_2\})' & \textrm{if } e_1{\notin}\alpha \textrm{ or } e_2{\notin}\alpha,\\
(\alpha{\setminus}\{e_1,e_2\})'e & \textrm{if } e_1,e_2{\in}\alpha. 
\end{array}\right.
\end{displaymath}

The t-simulation~(\ref{eq:par-atom}) is constructed as follows. Let us abstract the event names, i.e. $r^k$ would be a trace where each $r$ is the label of a unique event. Every trace of $r^{\bks}\|r'$ is a prefix of $r^mr'r^{n}\unity$ or $r^m\unity r'$, for some $m,n\geq 0$. Every prefix of either trace corresponds to a unique trace of $r{\bks} (r'{\cdot} r^{\bks})$. For instance, the maximal trace $r^m\unity r'$ is associated to the weakly maximal trace $r^mr'$ of $r{\bks} (r'{\cdot} r^{\bks})$. Figure~\ref{fig:messy-simulation} shows an explicit construction of the t-simulation.
\begin{figure}
\begin{small}
\begin{displaymath}
\xymatrix@C-=0.2cm{
&&\emptyset\ar[d]\ar[dl]\ar[drrr]\ar@{.>}[rrrrrrrr]&&&&&&&&\emptyset\ar[dl]\ar[dr]&&&\\
&\unity\ar[dl]\ar@{.>}@/^/[rrrrrrrrur]&\ar@{.>}@/^/[rrrrrrrrr]r\ar[d]\ar[dl]\ar[drr]&&&r'\ar@{.>}[rrrr]\ar[d]\ar[dr]&&&&r'\ar[dl]\ar[d]&&r\ar[dr]\ar[d]&&\\
\unity r'\ar@{.>}@/^/[rrrrrrrrru]&r\unity\ar[dl]\ar@{.>}@/^/[rrrrrrrrrru]&\ar@{.>}@/^/[rrrrrrrrrr]rr\ar[dl]\ar[d]\ar[dr]&&rr'\ar@{.>}@/_/[rrrrrrr]\ar[d]\ar[dr]&r'\unity&r'r\ar[d]\ar[dr]&&r'\unity&r'r\ar[dl]\ar[d]&&rr'\ar[d]\ar[dl]&rr\ar[dr]\ar[d]&\\
r\unity r'\ar@{.>}[rrrrrrrrrrru]&rr\unity\ar[dl]\ar@{.>}@/_/[rrrrrrrrrrru]&rrr\ar@{.>}@/_/[rrrrrrrrrrr]\ar[d]&rrr'\ar@{.>}@/^/[rrrrrrrrr]\ar[d]&rr'\unity&rr'r\ar[d]&r'r\unity&r'rr\ar[d]&r'r\unity&rr'r\ar[dl]\ar[d]&rr'\unity&rr'r\ar[dl]\ar[d]&rrr'\ar[d]&rrr\ar[d]\\
rr\unity r'&&\dots&\dots&&\dots&&\dots&r'rr\unity&\dots&r'rr\unity&\dots&\dots&\dots
}
\end{displaymath}
\end{small}
The ``obvious" arrows, such as an arrow from $r'\unity$ to $r'\unity$, have been left out to keep the picture clear.
\caption{The t-simulation from $r^{\bks}\|r'$ to $r{\bks}(r'{\cdot} r^{\bks})$.}\label{fig:messy-simulation}
\end{figure}

The Simulation~(\ref{eq:par-nondet} )is similar. Every trace of $r^{\bks}\|(b{\cdot}\EE{+}c{\cdot}\FF)$ is a prefix of $r^mb\alpha$ or $r^mc\beta$ or $r^m\unity b\gamma$ or $r^m\unity c\zeta$, where $\alpha{\in}\TT(r^{\bks}\|\EE)$, $\beta{\in}\TT(r^{\bks}\|\FF)$, $\gamma{\in}\TT(\EE)$, $\zeta{\in}\FF$ and $n\geq0$. Again, prefixes of the first two traces correspond to a unique trace of $r{\bks}(b{\cdot} (r^{\bks}\|\EE) {+} c{\cdot}(r^{\bks}\|\FF))$. The maximal trace $r^m\unity b\gamma$ is again mapped to the weakly maximal trace $r^mb\gamma$. Similarly for the fourth case. This indeed results in a t-simulation.

The Simulation~(\ref{eq:par-seq}) is constructed as follows. Every trace of $r^{\bks}\|(r'{\cdot}\EE)$ is a prefix of $r^mr'\alpha$ or $r^m\unity r'\beta$ for some trace $\alpha{\in}\TT(r^{\bks}\|\EE)$ and $\beta{\in}\TT(\EE)$. We continue as in the previous case.
\end{proof}

\Prop{pro:get-rid-of-par} is used mainly to interleave the right operand $r^{\bks}$ systematically with the internal structure of $\EE$, while preserving the simulation order. More precisely, these equations are applied to generate algebraic proofs for the reduction of one expression into another, where the occurrence of $\|$ is pushed deeper into the sub-expressions (and possibly removed).

\section{Probabilistic rely-guarantee conditions}\label{sec:prgc}

Our first task towards the extension of the rely-guarantee method to probabilistic systems is to provide a suitable definition of a rely condition that contains sufficient quantitative information about the environment and the components of a system. 

From a relational point of view, as in Jones' thesis~\cite{Jon81}, a guarantee condition expresses a constraint between a state and its successor by running the relation as a nondeterministic program. Therefore, it is important to know whether some action is executed atomically or whether it is split into smaller components. 
For instance, when run in the same environment, a probabilistic choice between $x{:=}x{+}1$ and $x{:=}x{-}1$ produced from an \texttt{if\dots then\dots else} clause may behave differently from an atomic probabilistic assignment that assigns $x{+}1$ and $x{-}1$ to $x$ with the exact same probability.

Without probability, a common example of a guarantee condition for a given program is the reflexive transitive closure with respect to $(\|)$ of the union of all atomic actions in that program~\cite{Hoa09a} which completely captures all possible ``effects" of the program. Such a closure property plays a crucial role in the algebraic proof of Rule~\ref{rule:rg-standard} is achieved through \Prop{pro:get-rid-of-par}. This construction was introduced by Jones~\cite{Jon81} and later refined by others~\cite{Din02,Jon12,Hoa09a}.

Non-probabilistic rely-guarantee conditions usually take the form $\rho^{\bks}$ for some binary relation $\rho$, defined on the state space of the studied program. The transitive closure of $\rho$ with respect to the relational composition $(\cdot)$ is usually a desirable property. To obtain a probabilistic guarantee condition from a relation $\rho{\subseteq}\Omega{\times}\Omega$, we construct a probabilistic program $r{\in} \H_1\Omega$  such that 
\[
	r(s) = \{\mu{\in}\D \Omega\ |\ \mu(\{s'\ |\ (s,s'){\notin} \rho\}) = 0\}.
\]
Equivalently, $r$ is the convex closure of $\rho$. The following proposition then follows from that construction.
\begin{proposition}\label{pro:transitive-convex-closure}
If a relation $\rho{\subseteq}\Omega{\times}\Omega$ is transitive, then the convex closure $r$ of $\rho$ satisfies $r{\cdot}(r{+}\unity)\refbyh r$.
\end{proposition}

\begin{proofsummary}
It follows immediately from the transitivity of $\rho$ and the definition of $r$.
\end{proofsummary}

\begin{proof}
Let $\rho$ be a transitive relation, $r$ its associated probabilistic program, $s{\in}\Omega$ a state and  $\mu{\in} [r{\cdot}(r{+}\unity)](s)$. We need to show that $\mu{\in} r(s)$. By definition of the sequential composition $(\cdot)$ (\Eqn{eq:6-sequential-H}), there exists $\nu{\in} r(s)$ and a deterministic program $f\refbyh (1{+}r)$ such that $\mu = f{\convol}\nu$. Let $u{\in}\Omega$ such that $(s,u){\notin}\rho$, we are going to show that $\mu(u) = 0$. We have:
\begin{align*}
\mu(u) = \sum_{t{\in}\Omega} f(t)(s)\nu(t) = \sum_{t{\in}\Omega\wedge (s,t){\in}\rho} f(t)(u)\nu(t) = \sum_{t{\in}\Omega\wedge (s,t){\in}\rho\wedge (t,u){\in}\rho}f(t)(u)\nu(t).
\end{align*}
The second equality follows from $\nu(t) = 0$ for every $(s,t){\notin}\rho$. Similarly, the last equality follows from $f(t)(u) = 0$ for $(t,u){\notin}\rho$. The last expression reduces to $\sum_{t{\in}\Omega\wedge (s,u){\in}\rho}f(t)(u)\nu(t)$, by transitivity of $\rho$, which is an empty sum because $(s,u){\notin}\rho$. Therefore, $\mu(u) = 0$ for every $(u,s){\notin}\rho$, that is $\mu{\in} r(s)$.
\end{proof}

The convex closure of a relation $\rho$, given in \Prop{pro:transitive-convex-closure}, sometimes provides a very general rely condition that is too weak to be useful in the probabilistic case. In practice, a probabilistic assignment is considered atomic and the correctness of many protocols is based on that crucial assumption. Hence the random choice and the writing of the chosen value into a program variable $x$ is assumed to happen instantaneously and no other program can modify $x$ during and in-between these two operations. Thus, probabilistic rely and guarantee conditions need to capture the probabilistic information in such an assignment.

\begin{example} Let $x$ be a (integer) program variable with values bounded by $0$ and $n$. Let us write $x{:=}\mathtt{uniform}(0,n)$ for the program that assigns a random integer between two integers $0$ and $n$ to the variable $x$. A probabilistic guarantee condition for that assignment is obtained from the probabilistic program $r$ that satisfies, for every integer $s\in[0,n]$,
 \begin{equation}\label{eq:example-rely} 
r(s) = \left\lbrace\mu\ \left| \ \mu(\{0,n\}) \geq\frac{1}{n{+}1}\right.\right\rbrace.  \end{equation} 
The condition $r$ specifies the convex set of all probabilistic deterministic programs whose atomic actions establish a state in $\{0,n\}$ with probability at least $\frac{1}{n{+}1}$. In particular, $r$ is an overspecification of  $x{:=}\mathtt{uniform}(0,x)$ where the rhs occurrence of $x$ is evaluated to the initial value of $x$. Since $r$ is transitive, it can prove useful to deduce quantitative properties of $(x{:=}\mathtt{uniform}(0,x))^{\bks}$.  \end{example}

In practice, constructing a useful transitive probabilistic rely-guarantee condition is difficult, but the standard technique is still valid: the strongest guarantee condition of a given program is the nondeterministic choice of all atomic actions found in that program.
\begin{definition}\label{def:rely}
A \emph{probabilistic rely} or \emph{guarantee condition} $R$ is a probabilistic concurrent program such that $R\|R \simref R$.
\end{definition} 

In particular, the concurrent program $r^{\bks} = r{\bks}1$ is a rely condition because 
\begin{equation}\label{eq:rely-closure}
r^{\bks}\|r^{\bks}\simref r^{\bks}
\end{equation}
holds in the event structure model (\Prop{pro:get-rid-of-par} \Eqn{eq:star-trans}). This illustrates the idea that a rely condition specifies an environment that can stutter or execute a sequence of actions that are bounded by $r$. 

\section{Probabilistic rely-guarantee calculus}\label{sec:prg-rules}

In this section, we develop the rely-guarantee rules governing programs involving probability and concurrency. An example is given by Rule~\ref{rule:rg-standard}, which allows us check the safety properties of the subsystems and infer the correctness of the whole system in a compositional fashion. We provide a probabilistic version of that rule.

In the previous sections, we have developed the mathematical foundations needed for our interpretation of Hoare triples and \emph{guarantee} relations, namely, the sequential refinement $\refby$ and simulation-based order $\simref$. Following~\cite{Hoa09a}, we only adapt the orders in the algebraic interpretation of rely-guarantee quintuples (\Eqn{eq:rgspec}). That is, validity of probabilistic rely-guarantee quintuples is captured by
\[
\triple{P\ R}{\EE}{G\ Q} \ \Leftrightarrow \ 	P{\cdot}(R\|\EE)\refby Q\, \wedge\, \EE\simref G,
\]
where $P,\EE$ and $Q$ are probabilistic concurrent programs and $R$ and $G$ are rely-guarantee conditions. The first part is seen as a probabilistic instance of the contraction of~\cite{Arm14} which specifies the functional behaviour of $R\|\EE$ under a precondition $P$. The second part uses the simulation order which is compositional and  very sensitive to the structural properties of the program.

The conditions $R$ and $G$ specify how the component $\EE$ interacts with its environment. As we have discussed in the previous section, rely and guarantee conditions are obtained by taking $r^{\bks} = r{\bks}\delta$ for some atomic probabilistic program $r$. Therefore, $\EE\simref r^{\bks}$ implies that all actions carried by events in $\EE$ are either stuttering or satisfying the specification $r$. This corresponds to the standard approach of Jones~\cite{Jon12,Jon81}.

The following rules are probabilistic extensions of the related rely-guarantee rules developed in~\cite{Hoa11,Jon12}. These rules are  sound with respect to the event structure semantics of Section~\ref{sec:es}.

\noindent{\textbf{Atomic action: }}
The rely-guarantee rule for an atomic statement $r'$ is provided by the equation
\begin{equation}\label{rule:atomic}
r^{\bks}\|r'\simref r{\bks}(r'{\cdot} r^{\bks})
\end{equation}
where $r$ is the rely condition. This equation shows that a (background) program satisfying the rely condition $r$ will not interfere with the low level operations involved in the atomic execution of $r'$. The programs will be interleaved.

\noindent{\textbf{Conditional statement:}}
The rely-guarantee rule for conditional statement is provided by the equation
\begin{equation}\label{rule:conditional}
r^{\bks}\|(b{\cdot} \EE {+} c{\cdot}\FF) \simref  r{\bks}(b{\cdot} (r^{\bks}\|\EE) {+} c{\cdot}(r^{\bks}\|\FF)).
\end{equation}
This equation shows how a rely condition $r^{\bks}$ distributes through branching structures. The tests $b$ and $c$ are assumed to be atomic and their disjunction is always \emph{true} (this is necessary for feasibility). This assumption may be too strong in general because $b$ may involve the reading of some large data that is too expensive to be performed atomically. However, we may assume that such a reading is done before the guard $b$ is checked and the non-atomic evaluation of the variables involved in $b$ may be assigned to some auxiliary variable that is then checked atomically by $b$.

\noindent{\textbf{Prefixing: }}
the sequential rely-guarantee rule for a probabilistic program expressed using prefixing. We have  
\begin{equation}\label{rule:prefix}
r^{\bks}\|(r'{\cdot}\EE)\simref r{\bks}(r'{\cdot}(r^{\bks}\|\EE)).
\end{equation}
It generalises Rule~\ref{rule:atomic} and tells us that a rely condition $r^{\bks}$ distributes through the prefixing operation. In other words, the program $r'$ and  $\EE$ should tolerate the same rely condition in order to prove any meaningful property of $r{\cdot}\EE$. This results from of our interpretation of $\|$ where no synchronisation is assumed. 

\noindent{\textbf{Concurrent execution:}}
in Rule~\ref{rule:rg-standard}, the concurrent composition $\EE\|\EE'$ requires an environment that satisfies $R{\cap} R'$ to establish the postcondition $Q{\cap} Q'$. However, such an intersection is not readily accessible at the structural level of event structures. Therefore, the most general probabilistic extension of Rule~\ref{rule:rg-standard} which applies to our algebraic setting is:
\begin{equation}\label{rule:rg-general}
\frac{\triple{P\ R}{\EE}{G\ Q}\qquad \triple{P\ R'}{\EE'}{G'\ Q'}\qquad G\simref R'\qquad G'\simref R}{\triple{P\ R''}{\EE\|\EE'}{G\|G'\ Q}},
\end{equation}
where $R''$ is a rely condition such that $R''\simref R$ and $R''\simref R'$.
The proof of this rule is exactly the same as in~\cite{Hoa11,Rab13}. In fact, we have $R''\simref R$, $\EE'\simref R$, $R\| R\simref R$, therefore \Eqn{eq:par-assoc} and Equational Implication~(\ref{eq:congruence}) imply 
\[
	R''\|(\EE'\|\EE)  \simref  R\|(R\|\EE)\simref R\|\EE,
\]
and we obtain $P{\cdot} R''\|(\EE'\|\EE) \simref P{\cdot} (R\|\EE)$ by \Eqn{eq:monotony}. It follows from \Thm{thm:trace-imply-distribution} that $P{\cdot} R''\|(\EE'\|\EE) \refby Q$. 

The conclusion does not contain any occurrence of $Q'$, but by symmetry, it is also valid if $Q'$ is substituted for $Q$. The combined rely condition $R''$ is constructed such that it is below $R$ and $R'$. Indeed, if $R,R'$ have a greatest lower bound with respect to $\simref$, then $R''$ can be taken as that bound, so that the strengthening of the rely is as week as possible. 

The above rule can be specialised by considering rely-guarantee conditions of the form $r^{\bks}$, where $r$ is an atomic probabilistic program. The following rule is expressed in exactly as in the standard case~\cite{Hoa11}. This is possible because probabilities are internal.

\begin{proposition}\label{pro:rule1}
The following rule is valid in BES:
\begin{equation}\label{rule:rg-atom-rely}
\frac{\triple{P\ r_1^{\bks}}{\EE_1}{g_1^{\bks}\ Q_1}\qquad \triple{P\ r_2^{\bks}}{\EE_2}{g_2^{\bks}\ Q_2}\qquad g_1\refbyh r_2\qquad g_2\refbyh r_1}{\triple{P\ (r_1{\cap} r_2)^{\bks}}{\EE_1\|\EE_2}{(g_1 {+} g_2)^{\bks}\ Q_1}},
\end{equation}
where $r,r',g,g'{\in}\H_1\Omega$ and $g{+}g'$ is the nondeterministic choice on $\H_1\Omega$.
\end{proposition}

\begin{proof}
This follows from substituting $R$ and $G$ by respectively $r^{\bks}$ and $g^{\bks}$ in Rule~\ref{rule:rg-general}. Moreover $g^{\bks}\|g'^{\bks}\simref (g{+} g')^{\bks}$ holds because $(g{+}g')^{\bks}\|(g{+}g')^{\bks}\simref (g{+}g')^{\bks}$ (\Eqn{eq:rely-closure}).
\end{proof}

Recall that the nondeterministic choice of $\H_1\Omega$ is obtained by the pointwise union followed by the necessary closure properties for the elements of $\H_1\Omega$. The intersection $r{\cap} r'$ is obtained by pointwise intersection.

\noindent{\textbf{Iteration:}}
a while program is modelled by using the binary Kleene star. The idea is to unfold the loop as far as necessary. The conditional and prefix (sequential) cases can then be applied on the unfolded structure to distribute the rely condition. That is, we write 
\begin{align*}\label{rule:while}
r^{\bks}\|((b{\cdot} \EE){\bks} c) & \simref r{\bks}(c{\cdot} r^{\bks} {+} b{\cdot} (r^{\bks}\|[\EE{\cdot}(b{\cdot}\EE{\bks} c)])).
\end{align*}
If $\EE$ is sequential, then $r^{\bks}$ can be ``interleaved" within the internal structure of $\EE{\cdot}(b{\cdot}\EE{\bks} c)$ by applying the prefixing and conditional statement rules.

The sequential correctness is achieved by the usual generation of probability distributions, obtained from terminating sequential behaviours, on the ``totally" unfolded event structure (assuming that $\EE$ is sequential). The sequential behaviours are usually obtained by interleaving the rely condition $r^{\bks}$ through the internal structure of the unfolded loop. A bounded loop, such as a for loop, should be modelled using a sequence of sequential compositions or prefixing.

\section{Application: a faulty Eratosthenes sieve}\label{sec:application}

In this section, we show how to use the previously established rely-guarantee rules to verify a probabilistic property of a faulty Eratosthenes sieve, which is a quantitative variant Jones' example~\cite{Jon81}.

Let $n\geq 2$ be a natural number and $s_0 = \{2,3,\dots,n\}$. For each integer $i$ such that $2\leq i\leq \sqrt n$, we consider a program $\thread_i$ that sequentially removes all (strict) multiples of $i$ from the shared set variable $s$ with a fixed probability $p$. More precisely, each thread $\thread_i$ is implemented as the following program: 
\[
	\begin{array}{lc}
		\texttt{for(j = 2 to n/i)} &  \\
		\qquad u_{i,j}:\  \texttt{skip}\ \pc{p}\ \texttt{remove(i*j from s)};&
	\end{array}
\]
where $\texttt{n/i}$ is the integer division of $\texttt{n}$ by $\texttt{i}$. Each $u_{i,j}$ can be seen as a faulty action that removes the product ${ij}$ from the current value of $s$ with probability $p$. The state space of each atomic deterministic program $u_{i,j}$ is $\Omega = \{ s\ |\ s{\subseteq} s_0\}$. In $\H_1\Omega$, $u_{i,j}$ is defined by $u_{i,j}(s) = (1{-}p)\delta_s {+} p\delta_{s{\setminus}\{ij\}}$. 
The whole system is specified by the concurrent execution 
\[
	\thread_2\| ... \|\thread_{\sqrt n}= \|_{i=2}^{\sqrt n}(u_{i,2}\cdots u_{i,\nicefrac{n}{i}}).
\]
where, in the sequel, $\sqrt n$ is computed without decimals.

Let  $\pi {=} \{2,3,\dots, m\}$  be the set of prime numbers in $s_0$. Our goal is to compute a ``good" lower bound probability that the final state is $\pi$, after executing the threads $\thread_i$ concurrently, from the initial state $s_0$.

We denote by $O_{i,j} = \{s\ |\ ij{\notin} s\}\subseteq\Omega$ and 
\[
	Q_{i,j}(s) = \{\mu{\in}\D\Omega \ |\ \mu(O_{i,j}){\geq} p\wedge \mu(\{s'\ | \ s'{\subseteq} s\}) {=} 1\}
\]
a specification of a probabilistic program that removes $ij$ from the state $s$ with at least probability $p$ and does not add anything to it. We define $O_i = {\cap}_{j=2}^{\nicefrac{n}{i}} O_{i,j}$, $Q_i = Q_{i,2}{\cdot} Q_{i,3}{\cdot}\dots{\cdot} Q_{i,\nicefrac{n}{i}}$ and $r$ to be the probabilistic program such that $r(s)$ is the convex closure of $\{\delta_{s'} \ |\ s'{\subseteq} s\}$. 

First, we show that every thread $\thread_i$ guarantees $r^{\bks}$. Second, we show that $\thread_i$ establishes $Q_i$ when run in an environment satisfying $r$, i.e. $r^{\bks}\|\thread_i\refby Q_i$, using the atomic and prefix rules~\ref{rule:atomic} and~\ref{rule:prefix}. Finally, we  apply the concurrency rule~\ref{rule:rg-atom-rely} to deduce that the system $\|_{i=2}^{\sqrt n}\thread_i$ establishes all postconditions $Q_2,Q_3,\dots Q_{\sqrt n}$, when run in an environment satisfying $r$.

\noindent{\textbf{Establising $\thread_i\simref r^{\bks}$ and $r^{\bks}\|\thread_i\refby Q_i$}}
 
On the one hand, it is clear that $u_{i,j}\refbyh r$, for every $i,j$, and thus $\thread_i\simref r^{\bks}$ follows from the unfold~(\ref{eq:unfold}). On the other hand, let us show that $r^{\bks}\|\thread_i\refby Q_i$. Multiple applications of the prefix-case give 
\[
	r^{\bks}\|\thread_i \simref r{\bks}(u_{i,2}{\cdot} (r{\bks} (u_{i,3}{\cdot}(\dots r{\bks} (u_{i,\nicefrac{n}{i}}{\cdot} r^{\bks}))))).
\]
Since the right multiplication $X\mapsto X{\cdot} r$, by any program $r{\in}\H_1\Omega$, is the lower adjoint in a Galois connection~\cite{Mci04}, the fixed point fusion theorem~\cite{Bac02} implies  
\[
	r{\bks}(u_{i,2}{\cdot} (r{\bks} (u_{i,3}{\cdot}(\dots r{\bks} (u_{i,\nicefrac{n}{i}}{\cdot }r^{\bks}))))) = r^{\bks}{\cdot} u_{i,2}{\cdot} r^{\bks}{\cdot} u_{i,3}{\cdot}\dots r^{\bks}{\cdot} u_{i,\nicefrac{n}{i}}{\cdot} r^{\bks},
\]
where the equality is in $\H_1\Omega$. Thus, 
\[
	r^{\bks}\|\thread_i \refby r^{\bks}{\cdot} u_{i,2}{\cdot} r^{\bks}{\cdot} u_{i,3}{\cdot}\dots{\cdot} r^{\bks}{\cdot} u_{i,\nicefrac{n}{i}}{\cdot} r^{\bks}
\]
follows from the fact that $\refby$ is weaker than $\simref$ (\Thm{thm:trace-imply-distribution}). The right hand side explicitly  states the interleaving of the rely condition $r^{\bks}$ in-between the atomic executions in $\thread_i$ as in~\cite{Jon12}. 

Moreover, since $r$ is the probabilistic version of a transitive binary relation, \Prop{pro:transitive-convex-closure} implies that $r{\cdot}(r{+}\unity)\refbyh r$. Since $\H_1\Omega$ is a probabilistic Kleene algebra~\cite{Mci05}, the right induction law of pKA implies $r^{\bks} = \unity {+} r$. This reduction of $r^{\bks}$ to $\unity{+}r$ illustrates the practical importance of transitive rely conditions. Therefore, \[
	r^{\bks}\|\thread_i \refby (\unity{+}r){\cdot} u_{i,2}{\cdot} (\unity{+}r){\cdot} u_{i,3}{\cdot}\dots (\unity{+}r){\cdot} u_{i,\nicefrac{n}{i}}{\cdot} (\unity {+}r),
\]
where the left hand side is a sequential program (thus \Prop{pro:homomorphism} enables us to use the definition of sequential composition of $\H_1\Omega$) directly
. Since $u_{i,j}\refby Q_{i,j}$, it remains to show that $(\unity{+}r){\cdot} Q_{i,2}{\cdot} (\unity{+}r){\cdot} Q_{i,3}{\cdot}\dots (\unity{+}r){\cdot} Q_{i,\nicefrac{n}{i}}{\cdot} (\unity {+}r)\refby Q_i$.

First we show that $Q_{i,j}{\cdot}(\unity{+}r)\refby Q_{i,j}$ and 
 $(\unity{+}r){\cdot} Q_{i,j}\refby Q_{i,j}$. Let $s{\in}\Omega$ and $\nu{\in} (Q_{i,j}{\cdot} (\unity{+}r))(s)$. By definition of the sequential composition in $\H_1\Omega$, there exists a probabilistic deterministic program $f\refbyh \unity {+}r$ and a distribution $\mu{\in} Q_{i,j}(s)$ such that 
$\nu(s') = \sum_{t{\in}\Omega}f(t)(s')\mu(t)$,
for every $s'{\in}\Omega$. Therefore, 

\[
	\nu (O_{i,j}{\cap}\{s'\ |\ s'{\subseteq} s\})  = \sum_{t{\in}\Omega}f(t)(O_{i,j})\mu(t) = \sum_{t{\subseteq} s}f(t)(O_{i,j})\mu(t),
\]
where the second equality follows from $\mu(\{t\ |\ t{\not\subseteq} s\}) {=} 0$, for every $\mu{\in} Q_{i,j}$. We deduce $\sum_{t{\subseteq} s}f(t)(O_{i,j})\mu(t)\geq p$, i.e. $\nu{\in}Q_{i,j}(s)$, by observing
\[
	\sum_{t{\subseteq} s}f(t)(O_{i,j})\mu(t)\geq \sum_{ij{\notin} t\wedge t{\subseteq} s}f(t)(O_{i,j})\mu(t) = \mu(O_{i,j}{\cap}\{t \ |\ t{\subseteq} s\})\geq p,
\]
because $f(t)(O_{i,j}) {=} 1$ for every $t$ such that $ij{\notin}t$ and $\mu(O_{i,j}{\cap}\{t\ |\ t{\subseteq} s\}) = \mu(O_{ij})$ for every $\mu{\in} Q_{i,j}(s)$. Consequently, $Q_{i,j}{\cdot}(\unity{+}r)\refby Q_{i,j}$.
Similarly, we can show that $(\unity{+}r){\cdot} Q_{i,j}\refby Q_{i,j}$ and thus $r^{\bks}\|\thread_i\refby Q_i$. 

\noindent{\textbf{Establising the property of $r^{\bks}\|_{i=2}^{\sqrt n}\thread_i$}}

Applying the rule~\ref{rule:rg-atom-rely} $\sqrt n{-}1$ times, we obtain, for every $Q_j$ such that $2{\leq} j{\leq} \sqrt n$,
\[
	r^{\bks}\|_{i=2}^{\sqrt n}\thread_i\refby Q_j.
\]
\noindent{\textbf{Inferring a lower bound for the probability of correctness}}

Unfortunately, Rule~\ref{rule:rg-atom-rely} does not give any explicit quantitative bound in term of probability for correctness. It does provide quantitative correctness, but all the probabilities are buried in the $Q_i$.

To obtain an explicit lower bound for the probability of removing all composite numbers, we first study the case of two threads that run concurrently. We know from Rule~\ref{rule:rg-atom-rely} that $r^{\bks}\|\thread_2\|\thread_3\refby Q_2$ and $r^{\bks}\|\thread_2\|\thread_3\refby Q_3$. Therefore, for every $\mu{\in}\sem{r^{\bks}\|\thread_2\|\thread_3}(s_0)$, we have $\mu(O_2)\geq p^{\nicefrac{n}{2}{-}1}$ and $\mu(O_3)\geq p^{\nicefrac{n}{3}{-}1}$ because there are $\nicefrac{n}{2}{-}1$ (resp. $\nicefrac{n}{3}{-}1$) multiples of $2$ (resp. $3$) in $[3,n]$ (resp. $[4,n]$). Therefore, $\mu(O_1{\cup} O_2) {+} \mu(O_2{\cap} O_3) = \mu(O_1) + \mu(O_2) \geq p^{\nicefrac{n}{2}{-}1} {+} p^{\nicefrac{n}{3}{-}1}$ and 
\begin{equation}\label{e1601}
	\mu(O_2{\cap} O_3) \geq p^{\nicefrac{n}{2}{-}1} {+} p^{\nicefrac{n}{3}{-}1} {-} 1.
\end{equation}
In the construction of the lower bound in \Eqn{e1601}, we have only used the modularity of measures and, therefore, it can be transformed into a more general rely-guarantee rule with explicit probabilities (\Prop{p1605}).

Given a subset $O{\subseteq}\Omega$ and $p{\in}[0,1]$, we write $\sem{\EE}(s_0)(O)\geq p$ if for every $\mu{\in}\sem{\EE}(s_0)$ we have $\mu(O)\geq p$. 

\begin{proposition}\label{p1605}
For every initial state $s_0$ and for all subsets $O_1,O_2{\subseteq}\Omega$,
\begin{footnotesize}
\begin{displaymath}
\frac{\sem{r_1^{\bks}\|\EE_1}(s_0)(O_1)\geq p_1\quad \sem{r_2^{\bks}\|\EE_2}(s_0)(O_2)\geq p_2\quad \EE_1\simref g^{\bks}\simref r_2^{\bks}\quad \EE_2\simref g'^{\bks}\simref r_1^{\bks}}{\sem{(r_1{\cap} r_2)^{\bks}\|\EE_1\|\EE_2}(s_0)(O_1{\cap} O_2)\geq p_1 {+} p_2 {-} 1\qquad \EE_1\|\EE_2\simref (g{+} g')^{\bks}}
.
\end{displaymath}
\end{footnotesize}
\end{proposition}

\begin{proof}
Let $\mu\in\sem{(r_1{\cap} r_2)^*\|\EE\|\EE_2}(s_0)$, we need to show that $\mu(O_1{\cap} O_2)\geq p_1{+}p_2{-}1$ with the above definition of $p_1$ and $p_2$. 
 
Let us define $Q_1$ to be the (single event) ipBES whose event is labelled by the probabilistic program $u_1$ such that $u_1(s_0) = \{\mu\ |\ \mu(O_1){\geq} p_1\}$ else $u_1(s) = \D\Omega$ for $s\neq s_0$. Similarly, we define $Q_2$. Then the premises imply $r_1^*\| \EE_1\refby Q_1$ and 
$r_2^*\|\EE_2\refby Q_2$. By \Prop{pro:rule1}, we have \[
	\sem{(r_1{\cap} r_2)^*\|\EE_1\|\EE_2}\refbyh \sem{Q_1}\qquad \textrm{ and }\qquad \sem{(r_1{\cap} r_2)^*\|\EE_1\|\EE_2}\refbyh \sem{Q_2}.
\]
Therefore $\mu(O_1)\geq p_1$ and $\mu(O_2)\geq p_2$. Modularity of finite measures implies that $\mu(O_1{\cap} O_2) {+} \mu(O_1{\cup} O_2) = \mu(O_1) {+} \mu(O_2)\geq p_1 {+} p_2$. Hence, $\mu(O_1{\cap} O_2)\geq p_1 {+} p_2 {-} \mu(O_1{\cup} O_2)\geq p_1 {+} p_2 {-} 1$ since $\mu(O_1{\cup}O_2)\leq 1$.

The simulation $\EE_1\|\EE_2\simref (g{+} g')^{\bks}$ is also clear from \Prop{pro:rule1}.
\end{proof}

We know from the above discussion that 
\[
	\sem{r^{\bks}\|\thread_2\|\thread_3}(s_0)(O_2{\cap} O_3)\geq p^{\nicefrac{n}{2}{-}1} {+} p^{\nicefrac{n}{3}{-}1} {-} 1.
\]
Applying \Prop{p1605} on $\thread_2\|\thread_3$ and $\thread_4$ yields
\[
	\sem{r^{\bks}\|\thread_2\|\thread_3\|\thread_4}(s_0)(O_2{\cap} O_3{\cap}O_4)\geq p^{\nicefrac{n}{2}{-}1} {+} p^{\nicefrac{n}{3}{-}1} {+}p^{\nicefrac{n}{4}{-}1}{-} 2.
\]
Thus $\sqrt n{-}1$ applications of \Prop{p1605} give 
\[
	\sem{r^{\bks}\|_{i=2}^{\sqrt n}\thread_i}(s_0)({\cap}_{i=2}^{\sqrt n} O_i) \geq \sum_{i=2}^{\sqrt n} p^{\nicefrac{n}{i}{-}1} {-} (\sqrt n{-}2) = f(p,n).
\]
The lower bound $f(p,n)$ sometimes provides a bad lower-approximation for the probability that the system establishes ${\cap}_{i=2}^{\sqrt n}O_i$. However, it is clear that $\lim_{p{\to} 1}f(p,n) = f(1,n) = 1$. 

In the particular case of $n=15$, we have $\sqrt{15} = 3$ and we only need to consider $\thread_2$ and $\thread_3$ so that $f(p,15) = p^6 {+} p^4 {-} 1$. The plot of $f(p,15)$ in \Fig{fig:comparison} shows that $f(p,15)$ gives a positive lower bound when $p\geq 0.868$, the exact probability being $p^{10} {+} 4p^9(1{-}p) {+} 4p^8(1{-}p)^2$. 

\noindent{\textbf{Refining the lower bound}}

We can use other internal properties of the system to obtain a better lower bound. It is clear that $O_i$ is an invariant for every $\thread_j$ (for $j{\neq} i$) and that all actions $u_{i,j}$ (sequentially) commute with each other. Thus, we should obtain a better lower bound by noticing that the system is ``sequentially better" than the following interleaving: $\thread_2$ removes all (strict) multiples of $2$, $\thread_3$ removes all multiples of $3$ assuming that all multiples of $\mathrm{lcm'}(2,3)$ (the lowest common multiple of $2$ and $3$ that is strictly greater than both) have been removed by $\thread_2$, and so on~\footnote{The probability of removing all composite numbers is  usually above that bound because $6$ can be removed by either $\thread_2$ or $\thread_3$.}. Thus 
\begin{align*}
	\sem{r^{\bks}\|_{i=2}^{\sqrt n}\thread_i}(s_0)({\cap}_{i=2}^{\sqrt n} O_i) &\geq p^{\nicefrac{n}{2}-1}p^{\nicefrac{n}{3}{-}1{-}[\nicefrac{n}{6}]}p^{\nicefrac{n}{4}{-}1{-}[\nicefrac{n}{4}{-}1]}p^{\nicefrac{n}{5}{-}1{-}[\nicefrac{n}{10}{+}\nicefrac{n}{15} - \nicefrac{n}{30}]}\cdots \\
	&= g(p,n),
\end{align*}
where the square-bracketed terms are the numbers of multiples remove by threads with smaller indices. For example, before $\thread_5$ runs, $\thread_2$ removes $\nicefrac{n}{10}$ multiples of $\mathrm{lcm'}(2,5)$, $\thread_3$ removes $\nicefrac{n}{15}{-}\nicefrac{n}{30}$ multiples of $\mathrm{lcm'}(3,5)$ (not multiples of $\mathrm{lcm'}(2,5)$), thus $\thread_5$ removes the remaining $\nicefrac{n}{5}{-}1{-}[\nicefrac{n}{10}{+}\nicefrac{n}{15}{-}\nicefrac{n}{30}]$ multiples of $5$. In the particular case of $n=15$, this yields
\[
	g(p,15) = p^{\nicefrac{15}{2}{-}1}p^{\nicefrac{15}{3}-1-\nicefrac{15}{6}} = p^{7-1+5-1-2} = p^8.
\] A graphical comparison of $f,g$ and the actual probability is displayed in Figure~\ref{fig:comparison} for $n = 15$.

\begin{figure}
\centering
\includegraphics[width=0.75\linewidth]{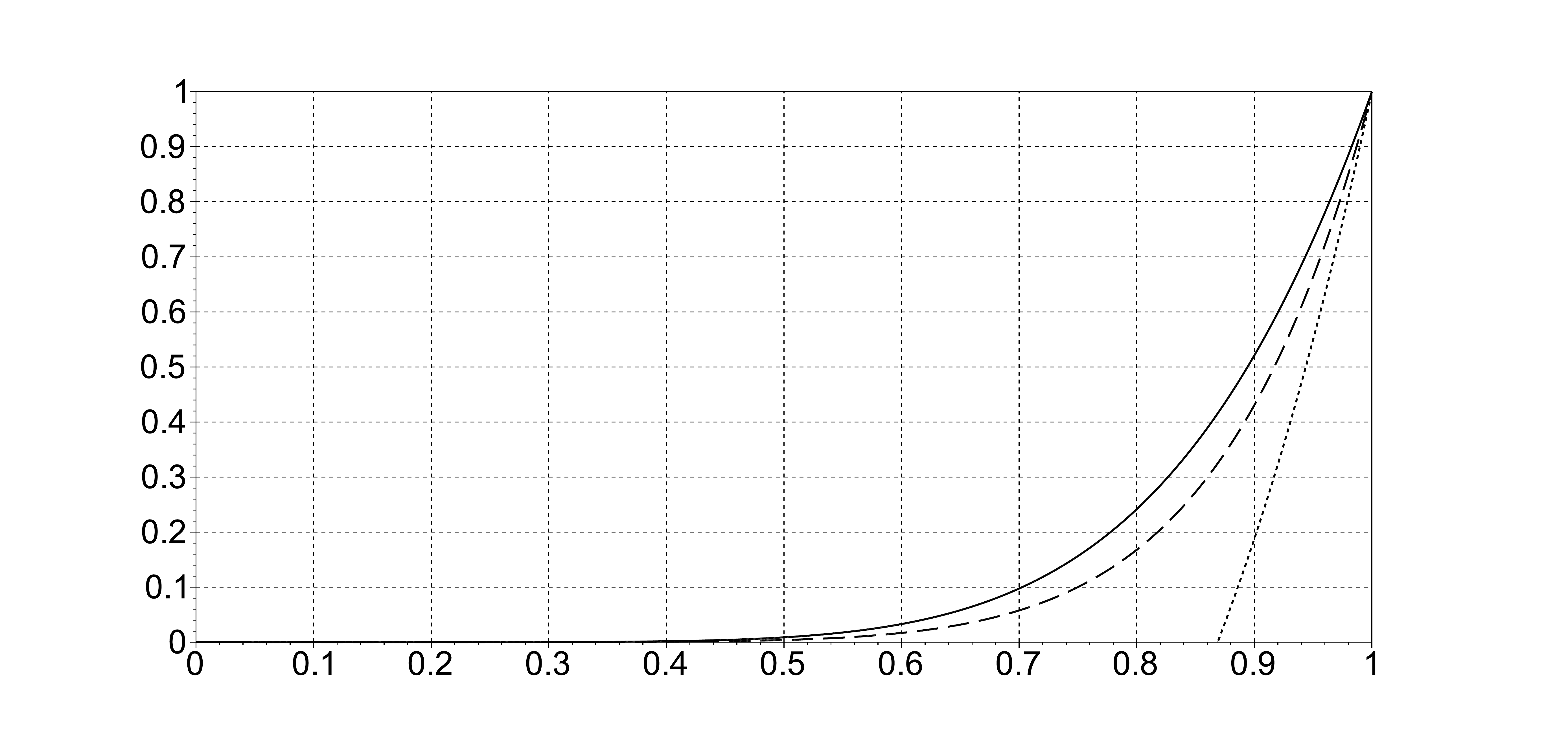}
\caption{Comparison of the quantities $f(p,15)$ (dotted), $g(p,15)$ (dashed) and the actual probability $p^{10} {+} 4p^9(1{-}p) {+} 4p^8(1{-}p)^2$ (solid).}
\label{fig:comparison}
\end{figure}

\noindent{\textbf{Establising the property of $\|_{i=2}^{\sqrt n}\thread_i$}}

Finally, notice that $\emptyset{\in} {\cap}_{i=2}^{\sqrt n}O_i$ which means that $r^{\bks}\|_{i=2}^{\sqrt n}\thread_i$ can establish $s=\emptyset$ with a positive probability. This issue is resolved by using a stronger guarantee property such as ``$u_{i,j}$ never removes $i$". Therefore, $\|_{i=2}^{\sqrt n}\thread_i$ never removes any prime numbers i.e. any element of ${\cap}_{i=2}^{\sqrt n}O_i$, that does not contain all the positive prime numbers below $n$, occurs with probability $0$.
 
\section{Conclusion}

We have presented an extension of the rely-guarantee calculus that accounts for probabilistic programs running in a shared variable environment. The rely-guarantee rules are expressed and derived by and large by using the algebraic properties of a bundle event structure semantics for concurrent programs. 

In our approach, the specification of a probabilistic concurrent program is expressed with a rely-guarantee quintuple. Each quintuple is defined algebraically through the use of a sequential order $\refby$, which captures all possible sequential behaviours when a suitable definition of the concurrency operation $\|$ is given, and a simulation order $\simref$, which specifies the level of interference between the specified component and the environment. Various probabilistic rely-guarantee rules have been established and applied on a simple example of a faulty concurrent system. We have also shown some rules that provide explicit quantitative properties, including a lower bound for the probability of correctness. In particular, a better lower-approximation can be derived if further internal properties of the systems are known.

The framework developed in this paper has its current limitations. Firstly, neither the algebra nor the event structure model support non-terminating probabilistic concurrent programs at the moment. That is, the rely-guarantee rules of this paper can only be applied in a partial correctness setting. Secondly, the concrete model is restricted to programs with finite state spaces. We will focus particularly on the first limitation in our future work.

\bibliography{./tcs-qapl}

\begin{append}
\newpage
\appendix
\section{Axioms of Kleene algebra and related structures}\label{A:ka}
\subsection{Idempotent semiring}
An \emph{idempotent semiring} is an algebraic structure $(K,+,\cdot,0,1)$ such that, for every $x,y,z{\in}K$, the following axioms hold
\begin{eqnarray}
x + x & = & x,\label{eq:+-idem}\\
x + y & = &y + x,\label{eq:+-comm}\\
x + (y + z) & = & (x + y) + z,\label{eq:+-assoc}\\
x + 0 & = & x, \label{eq:+-zero}\\
x{\cdot} 1 & = & x,\label{eq:rone}\\
1{\cdot} x & = & x, \label{eq:lone}\\
x {\cdot} (y {\cdot} z) & = & (x {\cdot} y) \cdot z,\label{eq:seq-assoc}\\
0{\cdot} x& = & 0, \label{eq:left-zero}\\
x{\cdot} 0& = & 0, \label{eq:righ-zero}\\
(x+y){\cdot} z & = & x{\cdot} z + y\cdot z,\label{eq:+-dist-seq} \\
x{\cdot} y+x{\cdot} z & = & x {\cdot} (y+z).\label{eq:+-rdist-seq}
\end{eqnarray}
\subsection{Kleene algebra}
A \emph{Kleene algebra} is an algebraic structure $(K,+,\cdot,^*,0,1)$ where $(K,+,\cdot,0,1)$ is an idempotent semiring and the Kleene star $(^*)$ satisfies Kozen's axioms:
\begin{eqnarray}
x^* & = & 1 + x{\cdot} x^*,\label{eq:*-unfold}\\
z+x{\cdot} y\leq y & \Rightarrow & x^*{\cdot} z\leq y,\label{eq:*-linduction}\\
z+y{\cdot} x\leq y& \Rightarrow &z{\cdot} x^*\leq y.\label{eq:*-rinduction}
\end{eqnarray}
The induction laws~\ref{eq:*-linduction} (resp. ~\ref{eq:*-rinduction}) implies that $x^*$ is the least fixed point of $\lambda y.1+x{\cdot} y$ (resp. $\lambda y.1 + y{\cdot}x$).
\subsection{Probabilistic Kleene algebra}
A \emph{probabilistic Kleene algebra} has the same signature as Kleene algebra but weakens the distributivity law~\ref{eq:+-rdist-seq} and the induction rule~\ref{eq:*-rinduction} to:
\begin{eqnarray}
x{\cdot} y+x{\cdot} z & \leq & x {\cdot} (y+z),\label{eq:+-subdist-seq}\\
z+y{\cdot} (x+1)\leq y& \Rightarrow &z{\cdot} x^*\leq y.\label{eq:*-rweakinduction}
\end{eqnarray}
\subsection{Concurrent Kleene algebra}
A \emph{concurrent Kleene algebra} is composed of a Kleene algebra $(K,+,\cdot,^*,0,1)$ and a commutative Kleene algebra $(K,+,\|,^{(*)},0,1)$ (i.e. $\|$ is commutative) linked by the interchange law:
\begin{eqnarray}
	(x\|y)\cdot(x'\|y')&\leq& (x\cdot x')\|(y\|y').
\end{eqnarray}
\end{append}

\end{document}